\documentclass[english,reprint,aps,pre,superscriptaddress]{revtex4-2}
 
\usepackage[T1]{fontenc}
\usepackage[utf8]{inputenc}
\usepackage[english]{babel}
\usepackage{mathrsfs}
\usepackage{bm}
\usepackage{amsmath}
\usepackage{amsthm}
\usepackage{amssymb}
\usepackage{graphicx}
\usepackage{enumerate}
\usepackage{siunitx}
\usepackage{mathtools}
\usepackage{enumitem}
\usepackage{amsfonts}
\usepackage{multirow} 
\usepackage{makecell}
\usepackage{array,booktabs,multirow}
\usepackage[table,dvipsnames,usenames]{xcolor}

\newlength{\mycolumnwidth}
\newcommand{\faintmidrule}{\arrayrulecolor{black!30}\midrule\arrayrulecolor{black}}

\newif\ifpn
\pntrue 

\ifpn

\else

\fi

\definecolor{darkblue}{rgb}{0.0, 0.0, 0.55}
\usepackage[unicode=true,bookmarks=false,breaklinks=false,pdfborder={0 0 1}]{hyperref}
\hypersetup{colorlinks,linkcolor=darkblue,citecolor=darkblue}
\newcommand{\newstuff}[1]{\textcolor{blue}{#1}}
\newcommand{\newnewstuff}[1]{\textcolor{purple}{#1}}
\newenvironment{newstuffenv}
    {\color{blue}}
    {}

\renewcommand{\newstuff}[1]{{#1}}
\renewcommand{\newnewstuff}[1]{{#1}}



\newtheorem{prop}{Proposition}

\usepackage{newtxtext}

\DeclareMathSymbol{\shortminus}{\mathbin}{AMSa}{"39}

\global\long\def\xx{\boldsymbol{x}}
\global\long\def\1{\mathbf{1}}
\global\long\def\e{\varepsilon}

\global\long\def\taua{\tau_{A}}
\global\long\def\taui{\tau_{I}}
\global\long\def\pr{{\mathcal{P}}}

\global\long\def\dim{w}
\global\long\def\str{s}
\global\long\def\smalldim{w}
\global\long\def\smallstr{s}
\global\long\def\scriptdim{w}
\global\long\def\scriptstr{s}

\global\long\def\hp{b}
\global\long\def\prm{{\mathscr{P}}}
\global\long\def\prmO{{\mathscr{P}}_{0}}
\global\long\def\inac{\varnothing}
\global\long\def\hpopt{\hat{b}}
\global\long\def\lambdaopt{\hat{\lambda}}
\global\long\def\gammaConst{ \gamma }

\global\long\def\Xstrin{X_{\inac\vert \scriptstr}}
\global\long\def\Xdimin{X_{\inac\vert\scriptdim}}

\global\long\def\multAT{\frac{\phi}{\tau\eta_r}}
\global\long\def\avgtau{\overline{\tau}}
\global\long\def\multATNTe{\frac{\phi_{\e}}{\eta_{r(\e)}}}

\allowdisplaybreaks

\begin{document}
\title{Information bounds production in replicator systems}
\author{Jordi Piñero}
\affiliation{ICREA-Complex Systems Lab, Universitat Pompeu Fabra, 08003 Barcelona,
Spain}
\author{Damian R. Sowinski}
\affiliation{Department of Physics and Astronomy, University of Rochester, Rochester, NY, 14627, USA}
\author{Gourab Ghoshal}
\affiliation{Department of Physics and Astronomy, University of Rochester, Rochester, NY, 14627, USA}
\affiliation{Department of Computer Science, University of Rochester, Rochester, NY, 14627, USA}
\author{Adam Frank}
\affiliation{Department of Physics and Astronomy, University of Rochester, Rochester, NY, 14627, USA}
\author{Artemy Kolchinsky}\email{artemyk@gmail.com}
\affiliation{ICREA-Complex Systems Lab, Universitat Pompeu Fabra, 08003 Barcelona,
Spain}
\affiliation{Universal Biology Institute, The University of Tokyo, Tokyo 113-0033, Japan}
\affiliation{Barcelona Collaboratorium, Wellington 30, 08005 Barcelona, Spain}

\begin{abstract}
Environmental fluctuations can shape replicator dynamics, with important consequences for both prebiotic and modern ecosystems. However, it remains unclear how simple replicators can acquire and use information about fluctuating environments, given that such information processing is often assumed to require sophisticated mechanisms for sensing and control. Here, we show that even simple replicator networks can increase productivity by exploiting environmental information in a functional way. Using a model of autocatalytic replicators in a flow reactor, we derive an information-theoretic decomposition of productivity, with separate contributions from environmental uncertainty, side information, and distribution mismatch. We derive optimal strategies and universal bounds on the benefit of information and compare our findings with existing work, including ``Kelly gambling'' in information theory.  By applying our framework to a model of real-world molecular replicators, we demonstrate the benefits of internal memory and propose an experimental setup for detecting functional information in a minimal chemical system. 
\end{abstract}

\maketitle

\section*{Introduction}
\label{sec:intro}

In the presence of environmental fluctuations, organisms acquire and use information about their environments in order to maintain and propagate themselves. 
In this sense, living systems are strikingly different from most nonliving systems, which may exhibit statistical correlations with their environments but do not use such correlations for functional purposes. 
The ability to use information in a functional manner has been termed `semantic'~\cite{kolchinsky2018semantic,sowinski2023semantic,ruzzante2023synthetic,godfrey-smithBiologicalInformation2007,sowinski2024information,Sowinski2024eDW}, `meaningful'~\cite{nehaniv2002meaningful,koonin2016meaning}, and `functional information'~\cite{collier2008information,hazen2007functional} in the literature. \newstuff{Simply put, functional information refers to statistical correlations that are necessary for a system to achieve functional outcomes~\cite{hazen2007functional,kolchinsky2018semantic}.}

Until now, functional information has been mostly considered in the context of modern organisms, which have sophisticated genetic~\cite{adami2002complexity,koonin2016meaning} and sensory~\cite{tu2013quantitative,palmer2015predictive,tkavcik2016information,mattingly2021escherichia} systems for acquiring and processing information. 
However, it is possible that functional information was already present in minimal prebiotic replicators~\cite{egbert2023behaviour}, given that environmental fluctuations (dry/wet cycles, seasons, diurnal oscillations, etc.) are thought to play a key role in the origin of life~\cite{damer2020hot,ianeselli2023physical}. 


To study functional information in minimal replicators, one may consider replicator dynamics in a flow reactor. This is a standard theoretical setting for investigating various kinds of systems, ranging from 
self-replicating  molecules~\cite{eigen1971selforganization,schuster1983replicator,iwasa1988free,von1993minimal,karev2010replicator,baum2023ecology,kolchinsky2024thermodynamics,despons2024structural} to microbial  communities~\cite{stephanopoulos1979growth,hsu1980competition,cushing1980two,pavlou1990coexistence,smithTheoryChemostatDynamics1995}. Moreover, such flow setups are  now being experimentally realized by chemists studying the origin of life and synthetic self-replication~\cite{von1986self,lee1996self,lincoln2009self,corbett2006dynamic,adamski2020self,vela2022collective,bandela2021primitive,mizuuchi2022evolutionary}. 

In this paper, we investigate the relationship between environmental fluctuations and information in minimal replicators. Our general setup is motivated by the example shown in Fig.~\ref{fig:scheme}. There is a population of replicators in a flow setup that can, under certain conditions, copy themselves from reactants and undergo exchange reactions, which convert between replicator types. 
The replicators are treated as simple self-replicating entities, and we do not assume they support combinatorial information storage (as in polymers). 
The replicators are exposed to a fluctuating environment, including active phases during which they grow and compete. A natural metric of interest is  \emph{productivity}, the average rate at which replicators flow out of the reactor. 

\begin{figure*}[ht]
    \centering
    \includegraphics[width=1.9\mycolumnwidth]{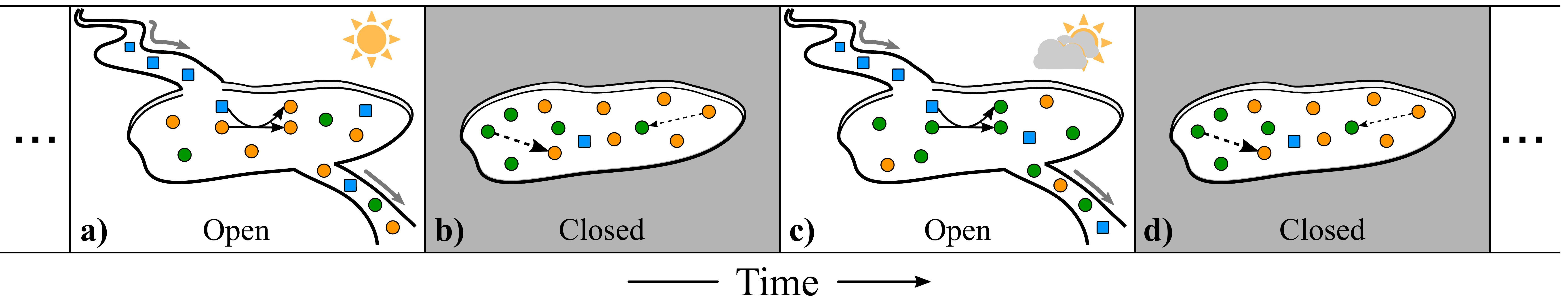}
    \caption{{\bf An example  scenario that motivates our analysis.}
    A reaction volume (e.g., pond) contains a network of replicators (e.g., orange and green circles), which may be minimal autocatalytic molecules or more complex organisms. \newstuff{During active phases {\bf(a,c)}, the system is open to flow and reaches a nonequilibrium steady state, while during inactive phases {\bf(b,d)}, the system is closed and replicator concentrations partially equilibrate. During activity (open phases), the replicators are also exposed to a fluctuating environment, e.g., days with different light intensities such as strong light {\bf(a)} or weak light \bf{(c)}}. 
    The replicators reproduce with different rates, which may depend on the state of the environment (light intensity), and they also undergo exchange reactions (dashed arrows). 
    We focus on \emph{productivity}, the amount of replicators per unit time that flows out of the system during active phases. As we show, productivity depends on the replicator concentrations at the beginning of each active phase, which may in turn depend on the kinetics of the exchange reactions that partially equilibrate the system during inactive phases. In this sense, the exchange reactions implement a \emph{strategy} for dealing with environmental uncertainty. 
    We show that productivity has an intuitive information-theoretic expression, reflecting contributions from overall uncertainty about the environmental state (weak or strong light), side information provided by internal memory, and mismatch between environmental statistics and the strategy implemented by the exchange reactions. We also connect this scheme to Kelly gambling, where inactive phases and active phases correspond to betting and gambling stages, respectively.}\label{fig:scheme}
\end{figure*}

We demonstrate that, in fluctuating environments, even simple replicators  can increase productivity by exploiting information about the environment. 
In particular, we derive a general expression for productivity 
that has 
contributions from three information-theoretic terms, reflecting in turn: 
{\em (i)} uncertainty about the current environment, {\em (ii)}  benefit of side information that helps predict the environment, and \emph{(iii)} mismatch between the actual and the optimal initial replicator proportions at the beginning of the active phases. 
The last term is the only contribution that depends on the initial replicator proportions, which we refer to as the \emph{strategy} (see Fig.~\ref{fig:scheme}). 
We identify the optimal strategy, i.e., the initial proportions that minimize mismatch (and thus maximize productivity) given the statistics of environmental fluctuations. 
Perhaps surprisingly, we find that optimal strategies are biased toward slower-growing replicators, reflecting a `head start' that compensates for the longer time they need to reach high productivity in a favorable environment.
After deriving our general theoretical results, we apply them to 
a real-world system: 
the photocatalytic replicators developed by Otto and collaborators~\cite{liu2024light}. Similar to the scheme shown in Fig.~\ref{fig:scheme}, 
this system is exposed to active cycles of weak and strong light as well as 
inactive phases, during which the replicator concentrations (partially) equilibrate due to exchange reactions. 
The inactive phases allow the system to establish a strategy for exploiting active phases, while possibly maintaining an internal memory. 
In temporally-correlated environments, this internal memory {of past environments} can serve as a source of side information that increases productivity. In the setting of this minimal replicator system, we suggest an experimental setup for detecting signatures of functional information (information leading to increased productivity).

\vfill

\section*{Results}

\subsection*{Theoretical results}
\label{sec:res1}

 {In this section, we derive our general information-theoretic expressions for productivity. These results concern replicator dynamics and productivity in an open flow reactor, for instance as might occur during the active phases {\bf (a)} and {\bf (c)} illustrated in Fig.~\ref{fig:scheme}. For convenience, a summary of relevant parameters and variables can be found in Table~\ref{table:glossary}.}

\begin{table}[ht]
\centering 
\begin{tabular}{lcc}    \toprule
\emph{Parameter} & \emph{Symbol} & \emph{Units} \\\midrule
 Inflow concentration of reactant   & $\mu$  & C  \\ 
 Dilution rate (inverse residence time)              & $\phi$  & T$^{ \shortminus 1}$\\
 Temporal duration & $\tau$ & T \\
 Replication rate of species $i$             & $\eta_i$  & C$^{\shortminus1}$T$^{\shortminus 1}$\\
 Fraction of time reactor is open & $\alpha$ & ---\\
   \midrule
\emph{Variable} &   &   \\\midrule
 Reactant concentration inside reactor   & $a$  & C  \\ 
 Replicator concentration of species $i$     & $x_i$  & C  \\
 Total replicator concentration     & $X$  & C  \\
 Total solute concentration     & $S$  & C \\
 Productivity    & $\pr$  & CT$^{\shortminus1}$  \\
 Productivity bound with side information    & $\prm$  & CT$^{\shortminus1}$  \\
 Productivity bound with no side-information \quad   & $\prmO$  & CT$^{\shortminus1}$  \\
 Environment outcome & $\e$ & --- \\
 Initial preparation outcome & $y$ & --- \\
  \bottomrule
\end{tabular}
\caption{{\bf Summary of parameters and variables for our general theoretical results.} Units notation: C  for mass concentration (mass per volume), T for time, --- for dimensionless. 
Throughout the text, steady-state values are indicated by a superscript star\,$^*$. Values of $\tau,\phi$ for environment $\e$ are indicated with subscripts as $\tau_{\e},\phi_{\e}$. 
\label{table:glossary}}

\end{table}

\subsubsection*{Setup}
\label{sec:setup}

We consider a well-mixed continuous-flow reactor with a dilution rate $\phi$ containing $n$ replicator species, indicated as $x_i$ for $i\in\{1,\ldots,n\}$ below. 
Species may represent either biological organisms (e.g., microbes) or abiotic chemical compounds (e.g., self-replicating molecules), though we typically imagine the latter. 
The reactor is also supplied with reactant species, indicated as $a$ below, a necessary resource for replication, which flows into the reactor at mass concentration $\mu$.

Each replicator copies itself via an autocatalytic reaction from reactants. Specifically, the concentration of replicator $i$ at time $t$ evolves as
\begin{align}
    \frac{d}{dt}x_i(t) &= \eta_i a(t)x_i(t) - \phi x_i(t)\,.\label{eq:replicator-dynamics} 
\end{align}
{The term  $\eta_i a(t)x_i(t)$ represents autocatalysis of the replicator, given reactant concentration $a(t)$ and rate constant $\eta_i$. The term $\phi x_i(t)$ represents outflow of the replicator from the reactor. 
We work with mass concentrations (mass per volume) throughout, so  our kinetic equations represent transport of mass, not counts.} 
The reactant concentration evolves as
\begin{align}
    \frac{d}{dt}a(t) &= \mu\phi - \sum_{i}\eta_i a(t) x_i(t) - \phi a(t).\label{eq:reactant-dynamics}
\end{align} 
The term $\mu\phi$ represents the inflow of reactant from the external source, $\sum \eta_i a(t) x_i(t)$ represents consumption of reactant by the replicators, and $\phi a(t)$ represents reactant outflow. \newstuff{Note that the reactor volume stays constant, so that the same volume enters and leaves the system per unit time. For this reason, we use the same dilution rate $\phi$ for both inflow in Eq.~\eqref{eq:reactant-dynamics} and outflow terms in Eqs.~\eqref{eq:replicator-dynamics} and~\eqref{eq:reactant-dynamics}.}

We assume that replication is 
first-order in reactant concentration $a(t)$. First-order kinetics of this kind have been observed in chemical replicators~\cite{robertson2014highly} and they are consistent with standard models of biological growth (e.g., Monod model) at low reactant concentrations~\cite[p.~43]{pepper2009environmental}. 
In addition, we do not include exchange reactions between different replicators in Eq.~\eqref{eq:replicator-dynamics} because we assume that, during active periods of replication, exchange reactions are much slower than autocatalytic growth. This does not preclude exchange reactions from becoming relevant during inactive periods, when the reactor is closed to inflows and outflows.

We define two useful quantities: the {\em total replicator concentration}, $X(t):=\sum_i x_i(t)$, and the {\em total solute concentration} $S(t):=X(t)+a(t)$.
Adding up Eqs.~\eqref{eq:replicator-dynamics}-\eqref{eq:reactant-dynamics} gives the dynamics of total solute concentration as
\begin{align}
\frac{d}{dt}S(t)=\phi\left(\mu-S(t)\right)\,,
\end{align}
which is solved by:
\begin{align}
    S(t) = S(0)e^{-\phi t}+\mu(1-e^{-\phi t}).
    \label{eq:solute-evolution}
\end{align}

Steady-state concentrations are indicated as $a^*$ and $x_i^*$ for the reactant and replicators, and $X^*$ and $S^*$ for the totals. 
In the generic case without neutrality (all $\eta_i$ are different), 
only one replicator can be present in steady state. 
Given that our system involves only a single reactant, this result corresponds to the well-known principle of `competitive exclusion' in ecology~\cite{armstrong1980competitive}. 
If all replicators are present in the initial population, the one remaining replicator in steady state is 
\begin{align}
    r = \operatorname*{argmax}_i \eta_i. 
\end{align}
In the following, we refer to the replicator species $r$ as the `winner'. 
The steady-state concentrations are given by
\begin{gather}
    a^*=\frac{\phi}{\eta_r}\quad\qquad S^*=\mu\\
    x^*_r =X^*=\mu-\frac{\phi}{\eta_r}\,\label{eq:xr-solution}
\end{gather}
as long as $x^*_r>0$ (no washout). To avoid washout, we assume that the parameters satisfy 
$\mu > \phi/\eta_r$. \newstuff{This condition states that the dilution rate is slower than
the maximum possible growth rate of the fittest species.}

\begin{figure}[t]
    \centering
    \includegraphics[width=\mycolumnwidth]{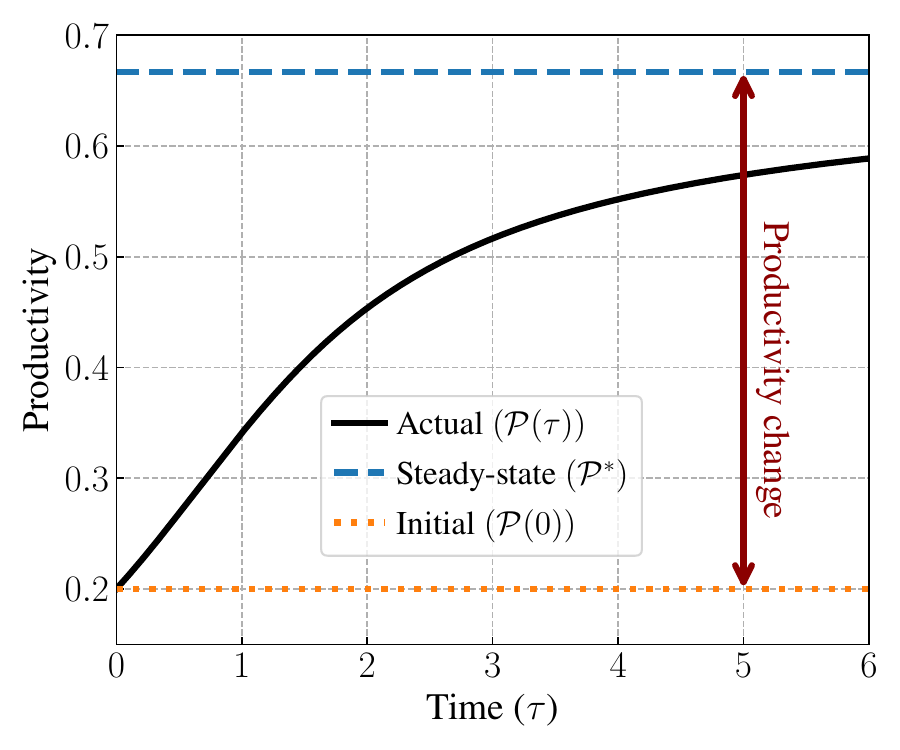}
    \caption{{\bf Productivity over time}. 
    Our result in Eq.~\eqref{eq:prod-result-1} is illustrated on a system of two replicators, with productivity defined in Eq.~\eqref{eq:productivity-def} as time-averaged outflow of replicators. Black curve is the actual productivity up to time $\tau$, orange dashed line is the initial productivity, and blue dashed line is the steady-state (long-time limit) productivity. 
     Red arrows indicate productivity change from initial  to steady-state values. 
    Parameters: $\eta_1 =2,\eta_2 = 3,\mu=1,\phi=1$; 
    initial concentrations: $a(0)=0.8\mu,x_1(0)= a(0)/5,x_2(0)=4a(0)/5$  ($S(0)=S^*=\mu$). For parameter definitions and units, see Table~\ref{table:glossary}.
}
    \label{fig:basics}
\end{figure}

\subsubsection*{Productivity}
\label{sec:productivity}

Suppose that the chemical system evolves over a time interval $t\in[0,\tau]$ from initial condition $\xx(0)=(x_1(0),\dots,x_n(0)),a(0)$. 
Our main quantity of interest is {\em productivity}, defined as:
\begin{align}
    \pr := \frac{1}{\tau}\int_0^{\tau} \phi X(t) \,dt. \label{eq:productivity-def}
\end{align}
Productivity is the time-averaged rate with which replicators flow out of the reactor, in units of mass concentration per time.  {In the motivating example shown in Fig.~\ref{fig:scheme}, \newstuff{productivity represents the rate at which replicators wash out of the pond during the active phases.} 
Since these replicators may potentially `infect' other ponds,  productivity can be imagined as a measure of fitness for replicators spreading between flow reactors.}

In the long-time limit $\tau \to \infty$,  productivity converges to its steady-state value, \begin{align}
\pr^*:=\phi x_r^*=\phi\Big(\mu-\frac{\phi}{\eta_r}\Big) \,.
\label{eq:prod-ss}
\end{align}
The assumption that $\mu>a^*$ guarantees that $\pr^*>0$. 
{To make things concrete, Fig.~\ref{fig:basics} shows  productivity over time for a simple system with two replicators. As replicator concentrations change, productivity approaches its steady-state value in the long-time limit.}

In what follows, we study how productivity $\pr$ depends on the initial concentrations $\xx(0)$ and $a(0)$. 
To do so, let us consider the winning replicator $r$. 
Dividing both sides of Eq.~\eqref{eq:replicator-dynamics} by $x_r(t)>0$ and integrating over $t\in[0,\tau]$ leads to
\begin{align}
    \ln \frac{x_r(\tau)}{x_r(0)} = \eta_r \int_0^{\tau} a(t)\,dt - \phi \tau \label{eq:r-prod-intermediate-step}.
\end{align}
Since $a(t)=S(t)-X(t)$ by definition of $S(t)$, 
\begin{align}
    \int_0^\tau a(t)\,dt = \int_0^\tau S(t)\,dt - \frac{\tau}{\phi}\pr,
\end{align}
where we used Eq.~\eqref{eq:productivity-def}. 
We integrate using Eq.~\eqref{eq:solute-evolution} and rearrange to obtain
\begin{align}
\pr = \pr^* +\frac{1-e^{-\phi \tau}}{\tau} \left(S(0)-S^*\right)+ \multAT\ln \frac{x_r(0)}{x_r(\tau)}, \label{eq:prod-result-1}
\end{align}
where $\pr^*$ is the steady-state productivity from Eq.~\eqref{eq:prod-ss}. 

Expression~\eqref{eq:prod-result-1} shows that the productivity
equals the steady-state productivity plus two correction terms.  
The first correction term in Eq.~\eqref{eq:prod-result-1} simply says that productivity increases in proportion to the excess initial solute $S(0)-S^*$ (i.e., the excess initial mass of replicator and reactants).

The last correction term in Eq.~\eqref{eq:prod-result-1} depends on the concentration change of the winner between the initial and final times, and it is a bit more subtle.  It implies that actual productivity is lower than $\pr^*$ (the steady-state productivity) when the winner's concentration increases, $x_r(\tau)>x_r(0)$. 
Intuitively, this means that replication events that increase the concentration $x_r$ within the reactor 
do not contribute to productivity (i.e., outflow). Conversely, actual productivity is larger than $\pr^*$ when $x_r(\tau)<x_r(0)$, reflecting excess initial  concentration of the winner that flows out as productivity, without having to be created by replication. 

In Fig.~\ref{fig:basics}, we provide a simple example of a system with two replicators.  {In this figure, the first correction term is zero because we set $S(0)=S^*$ for illustrative purposes; the second correction term is negative and leads to a gap between the steady-state productivity (blue dashed line) and the actual productivity (black curve). As we will see below, in cases where environmental fluctuations and internal relaxation have similar timescales, the contribution from this second term may have significant effects on long-term productivity.}

It will be useful to introduce two natural simplifying assumptions. 
{First, we assume that the initial solute concentration is approximately equal to its steady-state value:}
\begin{align}
S(0)\approx S^*=\mu\,.\label{eq:Sassump}
\end{align}
This assumption is valid for systems previously exposed to many active periods with the same dilution rate, possibly interspersed with closed periods (note that the amount of solute does not change when the reactor is closed), allowing the solute concentration to stabilize (see Fig.~\ref{fig:scheme}).
Second, we assume 
that the temporal duration $\tau$ is long enough so that the system approaches steady state, 
\begin{align}
    x_r(\tau)\approx x_r^*=X^*. \label{eq:long-tau-ss-assumption}
\end{align}
Combining these assumptions with Eq.~\eqref{eq:prod-result-1} gives 
\begin{align}
    \pr = 
    \pr^*+\multAT\ln \frac{x_r(0)}{X^*}. \label{eq:prod-result-2}
\end{align}
Importantly, while the time interval is taken to be sufficiently long so that the system reaches its steady-state value, there may still be a significant difference between the actual productivity $\pr$ and the steady-state productivity $\pr^*$, as quantified by the second term in Eq.~\eqref{eq:prod-result-2}.

The difference between $\pr$ and $\pr^*$ depends on the initial replicator concentrations. It can be further decomposed into two contributions: one due to the \emph{relative} amount of each replicator,  and one due to the \emph{total} amount of all replicators. We quantify the former using the {\em initial distribution} $q$, i.e., the normalized fraction of concentration (proportion of replicator mass) belonging to each replicator species $i$:
\begin{align}
q_i := \frac{x_i(0)}{X(0)}\,.  
\end{align}
Finally, we rewrite Eq.~\eqref{eq:prod-result-2} using this distribution as
\begin{align}
    \pr = \pr^*+\multAT\left[\ln q_r + \ln \frac{X(0)}{X^*}\right]. \label{eq:prod-result-3}
\end{align}
The term $\ln [{X(0)}/{X^*}]$ represents the productivity cost of increasing the mass of all replicators within the reactor (rather than flowing out), and it does not distinguish between different replicators. The replicator-specific term $\ln q_r$ represents the cost of having too little initial concentration on the winning replicator that eventually dominates the system. The meaning of the multiplicative factor $\phi/\tau \eta_r$ is discussed at the end of the following section.
Eq.~\eqref{eq:prod-result-3} serves as the basis of much of our analysis below.

\subsubsection*{Fluctuating environments}
\label{sec:fluctuatinv-env}
We now imagine that our system interacts with a fluctuating \emph{environment}, represented by the discrete random variable $E$. 
Each state of the environment, $\e$, occurs with probability $p_{\e}=p(E=\e)$, and it determines the replication rates $\{\eta_i\}_{i\in\{1,\ldots,n\}}$. This reflects the fact that different environments may favor different replicator species. In fact, the environment determines the winning replicator --- i.e., the species with the highest replication rate, which we indicate as $r(\e)$ --- and the steady-state concentrations $a_{\e}^*$ and $X_{\e}^*$.  For example, in terms of the scenario illustrated in Fig.~\ref{fig:scheme},  days with strong~(a) vs. weak~(c) light correspond to two different environments, which may favor different replicators (orange vs. green). As a matter of convention, we do not treat inactive phases without flow (such as the closed periods in Fig.~\ref{fig:scheme}b,d) as environments.

In principle, the dilution rate $\phi_{\e}$ and temporal duration $\tau_{\e}$ may also depend on the environment.  Although we typically consider environments with the same dilution and duration ($\phi_{\e}=\phi$ and $\tau_{\e}=\tau$ for all $\e$), one can generally imagine varying environmental durations (e.g., shorter vs. longer seasons, etc.) and flow rates. The fraction of time spent in environment $\e$ is given by $p_{\e}\tau_{\e}/\avgtau$, where we introduce the average temporal duration
\begin{align}
\avgtau  := \sum_{\e} p_{\e}\tau_{\e}\,.\label{eq:avgtaudef}
\end{align}

The environment cannot immediately affect the system's initial concentration vector $\xx(0)=(x_1(0),\dots,x_n(0))$ at $t=0$. 
However, we suppose that initial concentrations may depend on another discrete random variable $Y$, which we refer to as the  \emph{preparation}.  The preparation may represent any external variable (e.g., time of day, physical location of the reactor, previous interactions, etc.) that is correlated with the initial concentrations. For instance, in the scenario shown in Fig.~\ref{fig:scheme}, $Y$ could represent the sequence of past environments, since these may have an effect on the initial concentrations in the current environment. As another example,  $Y$ could represent different choices of initial concentrations imposed by a scientist in a laboratory setting.
We write the initial concentrations given preparation $Y=y$ as $x_{i}(0\vert y)$, and similarly for total concentration, $X(0\vert y)=\sum_i x_{i}(0\vert y)$.

The relative (normalized) initial concentration defines a conditional probability distribution:
\begin{align}
q_{i\vert y} := \frac{x_{i}(0\vert y)}{X(0\vert y)}\,.\label{eq:init-dist-cond-y}
\end{align}
In the following, we use the term {\em strategy} to refer to the conditional distribution $q$ defined in Eq.~\eqref{eq:init-dist-cond-y}.  This terminology is inspired by literature on Kelly gambling, where `strategy' refers to the  fractions of a finite resource (e.g., wealth or replicator concentration) allocated to different stochastic outcomes. 
In general, the strategy is determined by a number of parameters, which can either be intrinsic to the system or controlled by an external experimentalist. 
 
\newstuff{As an example, consider the scheme of Fig.~\ref{fig:scheme}. Here,} the strategy is implemented by the exchange reactions, which re-balance replicator concentrations during the inactive phases and thus determine replicator concentrations at the beginning of the active phases. 
The kinetics of the exchange reactions may depend on various factors (e.g., temperature of the pond,  presence of catalysts, etc.), implying that different ponds may implement different strategies. 
Furthermore, in that example, if the inactive 
\newstuff{time} is short enough so that the system does not reach complete equilibrium, then the previous environment (either weak or strong light) may influence the initial concentrations in the current active phase. In this case, the preparation $Y$ could simply represent the {\newstuff{state of the previous environment. This gives rise to an intrinsic memory about preceding active rounds that is encoded in the distribution $q$. We will return to this point below in our Application to a real-world system.} }

We use $q_{r(\e)|y}$ to indicate the relative initial concentration assigned to the winner in environment $\e$ given preparation $y$. 
Following Eq.~\eqref{eq:prod-result-3}, under environment $\e$ and preparation $y$, the productivity is given by:
\begin{align}
    \pr_{\e,y}=\pr_{\e}^*+\frac{\phi_{\e}}{\tau_{\e}\eta_{r(\e)}}\left[\ln q_{r(\e)\vert y}+\ln\frac{X(0\vert y)}{X_{\e}^*}\right].\label{eq:prod-e-y}
\end{align}
In this expression, $\pr_{\e}^*=\phi_{\e}x_{r(\e)}^*$ is the steady-state productivity in environment $\e$, where we applied Eq.~\eqref{eq:prod-ss}. \newnewstuff{We emphasize that this result holds at finite but sufficiently long durations $\tau_{\e}$, so that the system has time to approach the steady state.}

The probability of observing environment $\e$ and preparation $y$ is governed by the joint distribution $p_{\e,y}$. In addition, we also allow the possibility of inactive periods where the reactor is closed ($\phi=0$) and therefore there is no production;  {for instance, in the scenario shown in Fig.~\ref{fig:scheme}, this corresponds to the \newstuff{closed phases (b) and (d)}. \newnewstuff{We make no assumptions about the chemical reactions that occur during the inactive phases, only that the reactor is closed. In principle, it is even possible that some autocatalytic reactions continue during the inactive phases (this phenomenon occurs in the real-world photocatalytic system we consider below).}

The parameter $\alpha \in [0,1]$ indicates the fraction of time that the reactor is open. Combining,  
average productivity is given by 
\begin{align}
    \langle \pr\rangle=  \frac{\alpha}{\avgtau}\sum_{\e,y} p_{\e,y} \tau_{\e} \pr_{\e,y}\,.
    \label{eq:avgenvScale}
\end{align}
\newstuff{We note that $\langle \pr\rangle$ quantifies time-averaged productivity. For this reason, environments are weighted by their temporal duration $\tau_{\e}$ in Eq.~\eqref{eq:avgenvScale}, since the system spends more time in long-lived environments. (For more details, see~Methods.)}

We are interested in how much the expected productivity deviates from the expected steady-state productivity that would be reached in the limit of \newnewstuff{infinitely long} environments ($\tau_{\e} \to \infty$ for all environments), 
\begin{align}
    \langle \pr^*\rangle &=\frac{\alpha}{\avgtau}\sum_{\e}p_{\e} \tau_{\e}  \pr_{\e}^*\,.\label{eq:def-ss-av-prod}
\end{align}
The difference between $\langle \pr \rangle$  and $\langle \pr ^*\rangle$ quantifies the overall cost of environmental fluctuations. 
Combining the above  and rearranging 
leads to an expression of the expected productivity, our main theoretical result (see Methods  
for a step-by-step derivation):
\begin{align}
     \langle \pr \rangle =  \langle \pr^*\rangle -\gammaConst- \Omega \, C_{\pi,q}(R\vert Y) .
     \label{eq:main-result}
\end{align}
In this expression, we have introduced the constant 
\begin{align}
    \gammaConst :=\frac{\alpha}{\avgtau}\sum_{\e,y}p_{\e,y}\multATNTe\ln \frac{X^*_{\e}}{X(0\vert y)}, \label{eq:def-gamma}
\end{align}
and the joint distribution, $\pi_{R,Y}$, for winner replicator $R$ and preparation $Y$,
\begin{align}
\pi_{r, y} & :=\frac{\alpha}{\Omega\avgtau}\sum_{\e:r(\e)=r}p_{\e,y}\multATNTe\label{eq:pi-s}
\end{align}
with the normalization constant
\begin{align}
\Omega := \frac{\alpha}{\avgtau}\sum_{r}\sum_{\e:r(\e)=r}p_{\e}\multATNTe.\label{eq:omega-s}
\end{align}
The quantity $C_{\pi,q}(R\vert Y)$ is known as the `conditional cross-entropy' in information theory,
\begin{align}
   C_{\pi,q}(R\vert Y) = -\sum_{r,y} \pi_{r, y} \ln q_{r\vert y}\,. \label{eq:C-0}
\end{align}

The term $\gammaConst$ is an additive constant that does not depend on the strategy $q$. As seen in Eq.~\eqref{eq:def-gamma}, this constant depends on the ratio of the total replicator concentration in steady state, $X^*_{\e}$, versus in the beginning of the active phase, $X(0\vert y)$. A positive $\gammaConst$ implies a decrease in productivity (i.e., decreased outflow of replicators) due to newly created replicators building up inside the reactor, rather than flowing out. Conversely, a negative $\gammaConst$ implies an increase in productivity due to
existing replicators flowing out of the reactor, without having to be newly created.

The cross-entropy term $C_{\pi,q}(R\vert Y)$ is an information-theoretic measure  that can be understood as a measure of conditional uncertainty about environmental outcomes. It is the only term that depends on the strategy $q$. As we will see in the next subsection,  
the joint distribution $\pi$ 
specifies the optimal strategy that maximizes productivity among all possible $q$.

Finally,  
the normalization constant $\Omega$ in Eq.~\eqref{eq:omega-s} multiplies the information-theoretic term, acting as 
an `effective temperature' that converts between dimensionless informational quantities (in nats) and productivity (in units of concentration per time).
To understand its physical meaning, observe that $\Omega$ is the average of terms like $\phi_{\e}/\avgtau\eta_{r(\e)}$. 
For each environment $\e$, the dilution rate $\phi_{\e}$ determines how fast concentrations within the reactor flow out as productivity.  
The denominator $\avgtau \eta_{r(\e)}$ is proportional to the number of doublings of the fittest replicator during average duration $\avgtau$. 
Thus, the contribution from the information term --- the one that depends on the strategy $q$ --- is greater when dilution rates are high, and also when winning replicators are slow (fewer doublings). 
This reflects the fact that slower replicators are less able to recover from sub-optimal initial conditions, therefore the choice of the wrong strategy will have a greater cost in terms of lost productivity.

\subsubsection*{Information decomposition and the optimal strategy}

In this section, we use information-theoretic techniques to derive a closed-form expression for the \emph{optimal strategy}. Specifically, we find the $q$ that achieves the highest productivity averaged across all environments, $\langle \pr \rangle$ from Eq.~\eqref{eq:avgenvScale}. We also calculate the maximum average productivity achieved by this strategy.

To proceed, we note that the cross-entropy in Eq.~\eqref{eq:main-result} is a nonnegative measure of information-theoretic uncertainty. It can be decomposed into a sum of three  contributions:
\begin{align}
   C_{\pi,q}(R\vert  Y)&=H_{\pi}(R)-I_{\pi}(R;Y)+D\big(\pi_{R|Y}\Vert q_{R\vert Y}\big).\label{eq:C}
\end{align}
 
The first term is the Shannon entropy of the identity of the winning replicator $R$ under distribution $\pi$,
\begin{align}
    H_{\pi}(R) =  -\sum_{r} \pi_{r} \ln \pi_{r}.
\end{align}
It quantifies the average uncertainty about $R$, given environmental fluctuations. In our results, it captures  
the productivity cost of eliminating this uncertainty by `learning' the identity of the winner.  
The second contribution is (minus) the mutual information between the winner $R$ and the preparation $Y$ under distribution $\pi$,
\begin{align}
    I_{\pi}(R;Y) = H_{\pi}(R) - H_{\pi}(R\vert Y)=\sum_{r,y} \pi_{r, y} \ln \frac{\pi_{r\vert y}}{\pi_{r}}\,.
\end{align}
It quantifies the reduction in uncertainty about the winner $R$ provided by the initial preparation $Y$.  In our results, it captures the productivity benefit provided by the \emph{side information}  in the initial preparation.  
The third term is the Kullback-Leibler (KL) divergence between the actual strategy $q_{R\vert Y}$ and the conditional distribution  $\pi_{R\vert Y}$,
\begin{align}
       D\big(\pi_{R|Y}\Vert q_{R\vert Y}\big) = \sum_{r,y} \pi_{r, y} \ln \frac{\pi_{r\vert y}}{q_{r\vert y}} \,.
\end{align}
This nonnegative quantity reflects the distribution mismatch between the actual strategy and the optimal strategy specified by $\pi$. Due to this mismatch, productivity may be low even when the initial preparation provides a large amount of side information. In simple terms, the system's initial concentration may carry a great deal of information about the environment, but the system's dynamics may not be able to exploit this information to increase productivity.

Only the third KL term in Eq.~\eqref{eq:C} depends on the initial distribution $q$, and it reaches its minimum value of zero when the strategy $q_{R\vert Y}$ matches the conditional distribution $\pi_{R\vert Y}$. 
Thus, $\pi_{R\vert Y}$ represents the optimal strategy that maximizes productivity. 
However, in a setting where $q$ can only be manipulated by a limited set of control parameters, the optimal strategy $\pi$ is not always achievable. 
In our Application to a real-world system, we provide an example in which the optimal strategy $\pi$ is not achievable in general.

We note that the optimal strategy $\pi_{r,y}$ in Eq.~\eqref{eq:pi-s} has a contribution from $p_{\e,y}$, the frequency of environments and preparations. 
This recalls the {\em proportional betting} strategy, known to be optimal in Kelly's operational approach to information theory~\cite{kelly1956new,cover1999elements}. However, we note that, in our setting, the optimal strategy is also biased towards replicators that are slower (smaller $\eta_{r(\e)}$) and/or undergo higher dilution rates (larger $\phi_{\e}$). This bias toward slower replicators may appear counterintuitive at first. To unpack this, observe that when a slow replicator begins with a low concentration in a favorable environment, it takes a longer time to catch up to the steady-state productivity than a fast replicator, and thus it incurs a bigger loss of productivity (the gap shown in Fig.~\ref{fig:basics}).  Therefore, to avoid incurring this loss when presented with environments which favor slow replicators, the optimal strategy gives these replicators a `head start'.

\subsubsection*{Information-theoretic productivity bounds}

We may derive two useful bounds on the average productivity. First, from the nonnegativity of KL divergence, we have the inequality
\begin{align}
    \langle \pr \rangle \leq \langle \pr^*\rangle -\gammaConst -\Omega[H_{\pi}(R)-I_{\pi}(R;Y)]=:\prm,\label{eq:second-main-result}
\end{align}
This bound is achieved by the optimal strategy, therefore Eq.~\eqref{eq:second-main-result} can be understood as the unavoidable cost of uncertainty due to environmental fluctuations.  Furthermore, since  $H_{\pi}(R)-I_{\pi}(R;Y)=H_{\pi}(R\vert Y)\ge 0$, we can also derive the weaker inequality  
\begin{align}
    \langle \pr \rangle  \leq \langle \pr^*\rangle -\gammaConst .\label{eq:second-main-result-weaker}
\end{align}
This bound is achieved by the optimal strategy under the additional assumption of perfect side information, in essence when the initial preparation places all replicator mass on the correct  replicator.

At the other extreme, we may consider the case where the preparation provides no side information about the environment, $I_{\pi}(R;Y)=0$; for example, this occurs if the system always starts with the same initial concentrations.  In this case, 
Eq.~\eqref{eq:second-main-result} becomes
\begin{align}
    \langle \pr \rangle \leq \langle \pr^*\rangle-\gammaConst -\Omega \, H_{\pi}(R)=:\prmO \label{eq:third-main-result}
\end{align} 
This bound is achieved by the optimal strategy without side information, $\pi_R$.

It is interesting to compare the optimal bounds with and without side information, 
$\prm$ from Eq.~\eqref{eq:second-main-result} versus $\prmO$ from Eq.~\eqref{eq:third-main-result}. The difference between these two  bounds is
\begin{align}
    \prm - \prmO = \Omega \, I_{\pi}(R;Y)\,.
    \label{eq:miexcess}
\end{align}
Here, $ I_{\pi}(R;Y)$ emerges as the natural operational measure of the benefit of side information for replicator systems. Eq.~\eqref{eq:miexcess} is universal, in the sense that --- apart from its dependence on the scaling factor $\Omega$ --- it does not depend on any other chemical properties of the system, such as the steady-state productivity or the constant $\gammaConst$.

\subsection*{Application to a real-world system}
\label{sec:photocat}

In this section, we propose a connection between our previous information-theoretic findings and empirically measurable quantities in a plausible real-world experimental setup. 
Drawing inspiration from recent work in prebiotic chemistry~\cite{monreal2020emergence,liu2024light}, we study photocatalytic molecular self-replicators in a flow reactor,  which provide a realistic implementation of the scenario illustrated in Fig.~\ref{fig:scheme}. Here, we also show in concrete terms how an autonomous system can implement a strategy and maintain an internal memory that serves as a source of side information.

\begin{figure}
    \includegraphics[width=\mycolumnwidth]{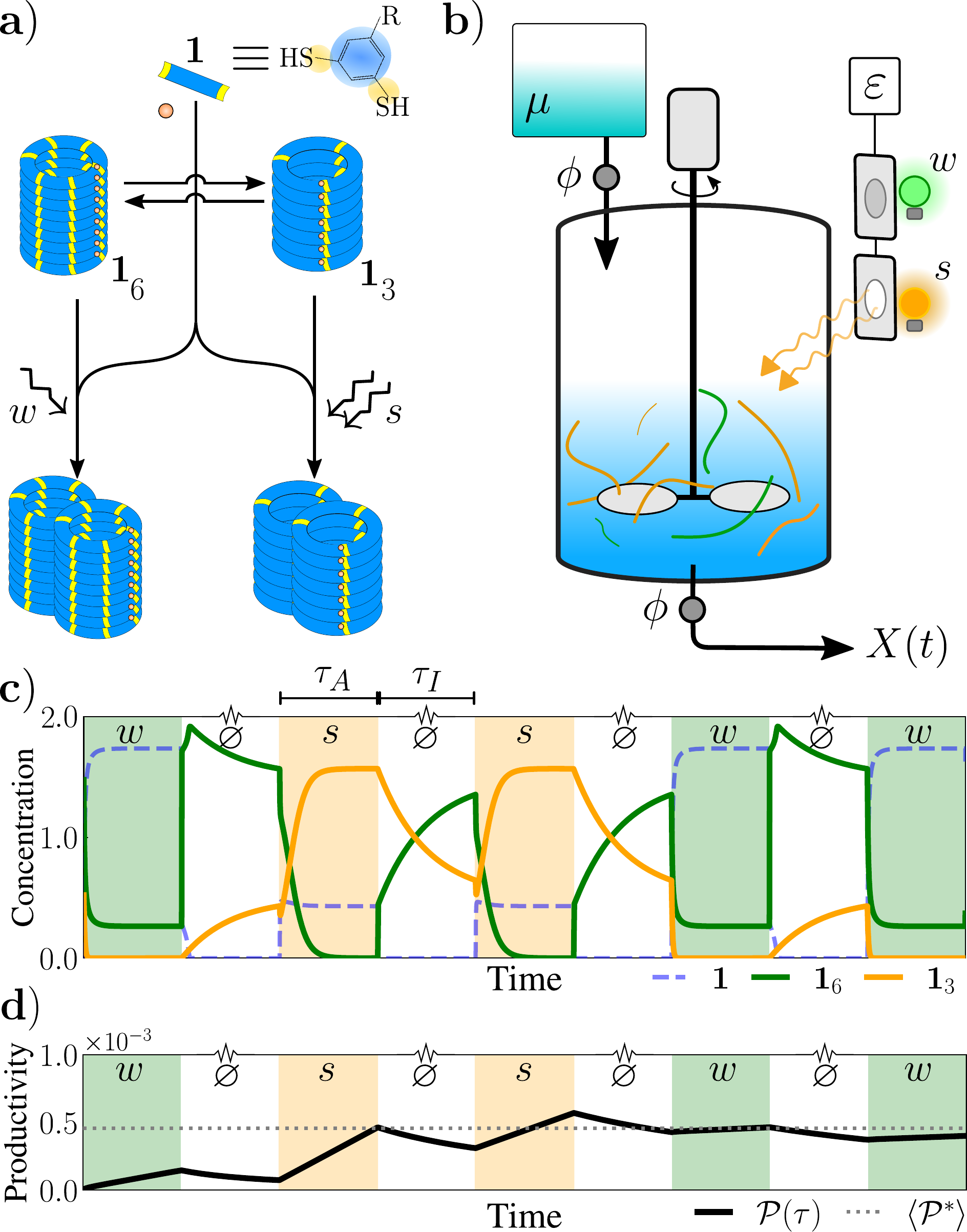}
    \caption{
    {\bf Simulated photocatalytic replicator system}.  
    \textbf{(a)} Schematic of simplified reaction network, where $\smalldim$ and  $\smallstr$  indicate replication in weak and strong light  environments, respectively (see Table~\ref{tab:reactions}). 
    \textbf{(b)} Experimental setup, including flow reactor fed by reservoir of monomers $\1$ at concentration $\mu$. 
    During active phases (weak $\smalldim$ or strong $\smallstr$ light environments), the reservoir feeds the reactor with rate $\phi$. 
    During inactive phase ($\inac$), light is switched off and the flow is stopped, allowing the system to establish a `bet' for the next active environment. 
    Productivity $\pr$ is quantified by measuring replicator concentration $X(t)$ at the outlet. \textbf{(c)} Typical concentration trajectories for monomer $\1$ and two replicators ($\1_6$ and $\1_3$), given cycles of inactive phases (white regions of length $\taui$) and randomly-chosen active environments (shaded regions of length $\taua$, with respective shaded colors green and orange for weak and strong light).  
    The timescale of the inactive phase is longer, but it is rescaled for illustrative purposes.
    Parameters used: $S(0)=S^*=\mu$, $\mu=2$, $\phi=5$, $\taua=20$, $\taui=2\times10^4$, 
    $\kappa=10^{-4}$ (giving $\lambda=2$) and $\hp=0.75$. See Table~\ref{tab:reactions} for parameter definitions and values of autocatalytic rates for $\1_6$ and $\1_3$ under different light conditions. \textbf{(d)} Productivity ($\pr$, black curve) over time from Eq.~\eqref{eq:productivity-def} and the steady-state productivity ($\langle \pr^*\rangle$, gray dotted line) from Eq.~\eqref{eq:ss-prod-photocat}.
    }
    \label{fig:setup}
\end{figure}

\subsubsection*{System and fluctuating environment}
\label{subsec:env-and-resetting}

In Ref.~\cite{liu2024light}, the authors demonstrated two self-replicating species of complex synthetic molecules (termed $\1_6$ and $\1_3$) representing hexameric and trimeric macrocycles that self-assemble  from a monomer species (termed $\1$). 
These macrocycles spontaneously stack to form respective fibers that catalyze their own production. 
Furthermore, by binding these replicators to photosensitive co-factors, the authors showed that these fibers enhance self-replication in response to different light stimuli. 
Under the right chemical conditions, replicator $\1_6$ wins in a {\em weakly lit} environment, and $\1_3$ wins in~a~{\em strongly lit} one. 
\newstuff{When the system is placed in {\em dark} conditions, autocatalysis occurs only for the hexamer $\1_6$ and at a smaller rate than with light. In addition to autocatalysis, there are also slower  
reactions that exchange matter between the two polymer populations.}
For details of the reaction scheme, see Fig.~\ref{fig:setup}a and Table~\ref{tab:reactions}.

We model this system as $n=2$ replicator species in a well-mixed reactor. 
We indicate the mass concentrations of $\1$, $\1_6$ and $\1_3$ as $a$, $x_6$, $x_3$, respectively, and use $X=x_6+x_3$ to indicate the total concentration of replicators. The reactor is coupled to an external cycle that turns light and flow {\em on} and {\em off}, which we denote as  {\em active} and {\em inactive} phases, respectively. 
Each active (open) phase has a duration $\taua$, \newstuff{during which the system approaches a nonequilibrium steady state}. During this time, the reactor is coupled to a reservoir containing reactant $\1$ at concentration $\mu$,  while inflow/outflow occurs with dilution rate $\phi$. In addition, the system is exposed to an environment with either \emph{weak light} (indicated as $\dim$), which favors replicator $\1_6$, or {\em strong light} (indicated as $\str$), which favors replicator $\1_3$. 
Each active phase is followed by an inactive phase of duration $\taui$ (indicated as $\inac$). 
During this time, the flow is closed off, light is removed (dark condition), and \newstuff{the system tends to (but does not necessarily reach) thermodynamic equilibrium. 
In this sense, our system alternates between two modes of control which have been previously termed `kinetic' (nonequilibrium) and `thermodynamic' (equilibrium), see Ref.~\cite{mattia2015supramolecular}.} 
Our basic setup is illustrated in Fig.~\ref{fig:setup}b.
The dynamics of the reactant and the two replicators during the active and inactive phases are described in more detail below, and are shown for illustration in Fig.~\ref{fig:setup}c. 

Importantly, the weak ($\dim$) or strong ($\str$) environments can exhibit temporal correlations.  For simplicity, we assume that the stochastic process over environments is stationary and first-order Markovian, and we use $p_{\e\vert \e_-}$ to indicate the conditional probability that a previous environment $\e_- \in \{\dim,\str\}$ is followed by the next environment $\e \in \{\dim,\str\}$ (always with an inactive phase in between). We use $p_{\e}$ to indicate the steady-state distribution of this Markov chain, which we assume has full support. We use  $p_{\e,\e_-}:=p_{\e\vert \e_-}p_{\e_-}$ to indicate the steady-state joint probability of environment $\e_-\in \{\dim,\str\}$ followed by environment $\e\in\{\dim,\str\}$. 
Therefore, we say that the environment random variable $E$ has two outcomes $\{\dim,\str\}$, which occur with respective probabilities $\{p_{\scriptdim},p_{\scriptstr}\}$. Note that the inactive phase is not treated as an environment since it does not contribute to outflow.

As we will see below, the environment during the previous active phase may influence the initial condition of the current active phase. For this reason, the previous environment $\e_-\in\{\dim,\str\}$ can serve as a source of side information $Y$. For notational convenience, we use the random variable $E_-\equiv Y$ to refer to the previous active environment. It also has two outcomes $\{\dim,\str\}$, which co-occur with the current environment $E$  according to the joint probability $p_{\e,\e_-}$.

We consider environments with different kinds of temporal correlations between  $\e$ and $\e_-$, as quantified by the sign of the correlation coefficient (see Methods)
\begin{align}
c:= p_{\scriptdim,\scriptdim} - p_{\scriptdim}p_{\scriptdim}=    p_{\scriptstr,\scriptstr} - p_{\scriptstr}p_{\scriptstr}\,\label{eq:twostatecov}
\end{align}
We say that environments are (positively) \emph{correlated} when \newstuff{$c>0$}, 
\emph{uncorrelated} when \newstuff{$c=0$}, and \emph{anticorrelated} when \newstuff{$c<0$}. 
In simple terms, the active condition (weak or strong) tends to repeat in correlated environments and to alternate in anticorrelated environments.

The production of replicators is tracked by measuring the outflow of both $\1_6$ and $\1_3$ at the outlet of the reactor. 
Since the reactor remains closed during inactivity, only active phases contribute to productivity. However, inactive phases allow the system to (partially) reset its state, thus setting up the initial condition for the subsequent active phase. As we will see below, the parameters of the exchange reactions (which still occur during inactive phases) affect the initial conditions of the subsequent active phase, thus implementing the system's strategy $q$.  We will classify the optimal strategies depending on whether the environments are correlated or anticorrelated.

Borrowing terminology from Kelly's original work~\cite{kelly1956new}: our inactive phase is interpreted as the `betting' stage, in which the system autonomously sets its strategy by preparing the initial condition for the next active round. 
Our active phase is akin to the `gambling' stage, in which the system evolves towards the  steady state dominated by the corresponding winner in the environment (light) state.

\subsubsection*{Reactions and dynamics} 
\label{subsec:reactions-dynamics-otto}
\begin{table}[t]
\centering
\begin{tabular}{lcc}    \toprule
\emph{Parameter} & \emph{Symbol} & \emph{Units} \\\midrule
 Duration of active phases (active time)   & $\taua$ & T\\
 Duration of inactive phases (inactive time) & $\taui$ & T\\
 Dimensionless inactive timescale ($\lambda=\kappa\taui$)\quad & $\lambda$ & ---\\
  \newstuff{Exchange rate constant for $\1_6\rightleftharpoons \1_3$} & \newstuff{$ \kappa $} & T$^{\shortminus1}$ \\
 \newstuff{Bias in favor of hexamer $\1_6$ in exchange reaction} & $b$ & --- \\
 \bottomrule
\end{tabular}

\vspace{7.5pt}

\renewcommand{\arraystretch}{1.5}
\setlength{\arrayrulewidth}{0.4pt}

\begin{tabular}{c|c}
\specialrule{\heavyrulewidth}{0pt}{0pt}
\begin{minipage}{4.5cm}\centering
\newstuff{\textit{Active/open phases $\left(\e\in\{\dim,\str\}\right)$}}
\end{minipage}
&
\begin{minipage}{3.25cm}\centering
\newstuff{\textit{Inactive/closed phases}}
\end{minipage}
\\
\hline

\begin{minipage}{4.5cm}\centering
$\1+\1_6 \overset{\newstuff{\eta_{6,\e}}}{\longrightarrow} \1_6+\1_6$
\end{minipage}
&
\multirow{3}{*}{
\begin{minipage}{3.25cm}\centering
\newstuff{$\1+\1_6 \overset{\eta_{6,\varnothing}}{\longrightarrow} \1_6+\1_6$}
\end{minipage}
}
\\

\begin{minipage}{4.5cm}\centering
$\1+\1_3 \overset{\newstuff{\eta_{3,\e}}}{\longrightarrow} \1_3+\1_3$
\end{minipage}
&
\\

\begin{minipage}{4.5cm}\centering
$\emptyset \overset{\mu\phi}{\longrightarrow} \1$ and
$\1,\1_6,\1_3 \overset{\phi}{\longrightarrow} \emptyset$
\end{minipage}
&
\\

\hline
\multicolumn{2}{c}{
\rule{0pt}{3.5ex}
\newstuff{$ \1_6 \overset{\kappa(1-b)}{\longrightarrow} \1_3$
and
$ \1_3 \overset{\kappa b}{\longrightarrow} \1_6$}
}
\\
\specialrule{\heavyrulewidth}{0pt}{0pt}

\end{tabular}

\vspace{7.5pt}

\centering

\renewcommand{\arraystretch}{1.0}
\begin{tabular}{lccc}
    \toprule
    \newnewstuff{\emph{Experimental values}} & \newnewstuff{\emph{Symbol}} &\newnewstuff{\emph{Value(s)}}  & \newnewstuff{\emph{Units} } \\
    \midrule
    \multirow{2}{*}{\newnewstuff{Autocatalytic rates in weak light}} 
        & \newnewstuff{$\eta_{6,w}$} & \newnewstuff{2.88} & \newnewstuff{C$^{\shortminus1}$T$^{\shortminus1}$} \\
        & \newnewstuff{$\eta_{3,w}$} & \newnewstuff{0.00} & \newnewstuff{C$^{\shortminus1}$T$^{\shortminus1}$}\\
    \faintmidrule
    \multirow{2}{*}{\newnewstuff{Autocatalytic rates in strong light} } 
        & \newnewstuff{$\eta_{6,s}$} & \newnewstuff{10.2} & \newnewstuff{C$^{\shortminus1}$T$^{\shortminus1}$}\\
        & \newnewstuff{$\eta_{3,s}$} & \newnewstuff{11.2} & \newnewstuff{C$^{\shortminus1}$T$^{\shortminus1}$}\\
    \faintmidrule
    \multirow{2}{*}{\newnewstuff{Autocatalytic rates in inactivity}} 
        & \newnewstuff{$\eta_{6,\varnothing}$} & \newnewstuff{1.76} & \newnewstuff{C$^{\shortminus1}$T$^{\shortminus1}$}\\
        & \newnewstuff{$\eta_{3,\varnothing}$} & \newnewstuff{0.00} & \newnewstuff{C$^{\shortminus1}$T$^{\shortminus1}$}\\
    \faintmidrule
    \newnewstuff{Exchange rate constant (range)} & \newnewstuff{$\kappa$} & \newnewstuff{5$\cdot$10$^{\shortminus7}-$10$^{\shortminus3}$} & \newnewstuff{T$^{\shortminus1}$} \\
    \bottomrule
\end{tabular}

\caption{ \textbf{Summary of parameters and experimental values for photocatalytic replicator model.}
\emph{Top}: Parameters used for model of photocatalytic replicators. T for units of time.
\emph{Middle}: Simplified reaction network for replicators $\1_6$ (replicator index $i=6$) and $\1_3$ (replicator index $i=3$), which self-assemble from reactant $\1$ (monomers). 
\newstuff{Active environments with weak ($\dim$) or strong ($\str$) light conditions lead to self-replication of $\1_6$ or $\1_3$, respectively. During activity, $\1$ flows into the system and all molecules outflow at the same rate. During the inactive phase, the system is closed off and only $\1_6$ can self-replicate. Exchange reactions (bottom row) occur throughout and at a slow rate. \newnewstuff{\emph{Bottom:} Summary of experimental values for autocatalytic and exchange rates extracted from~\cite{liu2022out,liu2024light}.}}
\label{tab:reactions}}

\end{table}

We describe the chemical reaction network introduced in~\cite{liu2024light} by a coarse-grained set of reactions summarized in Table~\ref{tab:reactions} middle. \newstuff{
The system has autocatalytic reactions that depend 
on the light intensity, whether dark (no light), weak, or strong. During open (active) phases, the system undergoes inflow/outflow and is exposed to an environment $\e$ with either weak $\dim$ or strong $\str$ light. During closed (inactive) phases, it is placed in the dark. 
Throughout, the system experiences slow exchange reactions (last row in the middle Table~\ref{tab:reactions})}, 
which effectively re-balance the concentrations of the two replicators.

The experiments in Ref.~\cite{liu2024light} used a flow reactor with residence time 16.7~h and inflow reactant concentration 0.5~mM. 
Without loss of generality, below we measure time and mass concentration as multiples of these values. 
For reasons discussed below, in our model, we use a somewhat higher inflow reactant concentration 1~mM ($\mu=2$ in our units) and shorter residence time $3.34\,\text{h}=16.7\,\text{h}/5$ ($\phi=5$ in our units).

\newstuff{For the autocatalytic rate constants $\{\eta_{6,\e},\eta_{3,\e}\}$ with $\e\in\{\dim,\str\}$, we use the growth rates reported in Figure 35 of the Supplemental Material in~\cite{liu2024light}. At 0.2~mM initial monomer concentration, $\1_6$ replicators grew $\approx~$7\% over one hour in weak light and $\approx~$27\% over one hour in strong light, while $\1_3$ replicators grew $\approx~$30\% over one hour in strong light and did not grow in weak light. In addition, Liu~\emph{et~al.} reported that replicator $\1_6$ undergoes slow autocatalysis under dark conditions ($\varnothing$), corresponding to a nonzero autocatalytic rate constant $\eta_{6,\varnothing}>0$~\cite{liu2022out,liu2024light}.} All autocatalytic rate constants are given in Table~\ref{tab:reactions} bottom. Note that we converted percentage increases $r$ to per-capita growth rate $g$ via $g=\ln(1+r/100)$, divided by the stated monomer concentration to get the rate constants $\eta$, and expressed them in the units of time and concentration introduced above. From these rate constants, we observe that the winning replicators ($R$) and the environments ($E$) are in a one-to-one relation:  $\1_6$ wins when $\e=\dim$ and $\1_3$ wins when $\e=\str$. This implies that the two random variables are equivalent, i.e., $R\equiv E$.

\newstuff{Next, we discuss exchange reactions. We emphasize that our analysis requires only that the concentrations of different replicators undergo partial equilibration during the inactive phase, and it does not depend on the precise mechanism through which this occurs. In practice, we suppose that $\1_3$ converts into $\1_6$ at a rate of $\hp \kappa$ and $\1_6$ converts into $\1_3$ at a rate of $(1-\hp) \kappa$, where $\hp \in [0,1]$ is the {\em bias} in favor of replicator $\1_6$ and $\kappa > 0$ is the exchange rate constant. These reactions represent effective exchange due to various mechanisms, including degradation and formation, cross-catalysis, and disulfide exchange. In principle, these reactions also affect the active phases, but their influence is negligible so long as they occur at much lower rates than autocatalysis. 
This rate constant may be controlled through various methods, such as varying the reduction agents in solution (such as TCEP) and/or stirring. 
We explore a range of exchange rate constants $\kappa$ from $5\times 10^{-7}$ to $10^{-3}$ (in the units described above). 
For the bias parameter, we do not select one particular value, but rather explore a range of bias values $\hp\in[0,1]$. 
The bias toward $\1_6$ or $\1_3$ can be controlled through oxidation level (see Fig.~3a in~\cite{liu2024light}). 
}

To allow slow exchange reactions to (partially) equilibrate concentrations during inactivity, we usually assume that the inactive timescale is much longer than the active one ($\taui\gg\taua$). 
This can be imagined as periodic bursts of activity followed by long relaxation (inactive) periods. It is useful to characterize the inactive phase by a {\em dimensionless inactive timescale} defined as:
\begin{align}
\lambda := \kappa \taui.\label{eq:dimensionless_timescale}
\end{align}
In simple terms, $\lambda$ is the number of \newstuff{exchange} events during the inactive phase per replicator.

\newstuff{We select  $\taua=20$ as the duration of the active phase, 
which ensures that the system approaches steady state within each active phase.} 
This implies that, at the end of any active phase, all the remaining dependence on the previous history is erased. Hence, the subsequent inactive phase will only depend on the previous active environment $\e_-$. \newstuff{The duration of the inactive phase is chosen as $\taui=2\times 10^4$. }%
\newnewstuff{As we show below, this choice of $\taui$  allows partial equilibration of the system during inactive phases.} 
Given the value of $\taui$ and Eq.~\eqref{eq:dimensionless_timescale}, the exchange rate constants $\kappa \in [5\times 10^{-7},10^{-3}]$ correspond to dimensionless inactive timescales $\lambda\in[10^{-2},20]$.

We comment on a few differences between the parameter values we use in our analysis and those reported in Ref.~\cite{liu2024light}. As mentioned above, we use a shorter residence time (3.34~h versus 16.7~h, corresponding to $\phi=5$ versus $\phi=1$). This is done to more clearly show numerical productivity differences that arise from the choice of strategy (observe that in the limit of long residence times, the choice of strategy matters less and less). We also increase the monomer inflow concentration (1~mM versus 0.5~mM, corresponding to $\mu=2$ versus $\mu=1$) to avoid washout at the higher dilution rates. In addition, our exchange rate constants are likely slower than the real-world ones. In particular, experimental data from closed reactors (Figs.~36 and 37 in the Supplemental Material of~\cite{liu2024light}) suggest an exchange timescale of roughly 5--15 days, which would correspond to  $\kappa\in [0.05,0.14]$ (in our units). However, for such large exchange rates, exchange reactions begin to interfere with replicator dynamics during the active phases.  Our analysis assumes that it is possible to increase the timescale separation between replication and exchange reactions during the active phases.

In the Methods section, 
we {compute} replicator concentrations at the end of an inactive phase, expressing them as a function of $\{\hp,\lambda\}$ and the previous environment $\e_-$. 
These concentrations serve as the initial conditions of the subsequent active phases, and their relative proportions determine the strategy $q$, see Eq.~\eqref{eq:init-dist-cond-y}. 
As discussed above, $\e_{-}$ enters in $q$ as a variable that contains side information about the environmental fluctuations. In other words, the strategy is characterized by $q_{R\vert E_-}$, which itself is a function of $\{\hp,\lambda\}$.
\newstuff{The timescale $\lambda$ regulates how much the strategy captures about the previous active environment. For small $\lambda$, the vector $q_{R \vert E_-}$ is highly \newnewstuff{dependent on} the preceding environment $E_-$; for larger  $\lambda$, $q_{R \vert E_-}$ becomes independent of $E_-$. In this way, $\lambda$ controls the degree of memory that the system retains about the preceding environment. \newnewstuff{We note that the amount of memory in the system is separate from environmental correlations, as quantified by the coefficient $c$ from Eq.~\eqref{eq:twostatecov}}.}

The connection between the strategy and side information $\e_-$ can also be interpreted as an intrinsic {\em first-order memory} of the system. The memory is first-order because it only depends on the last environment, being  reset by the end of every active phase (for a visual example, see Fig.~\ref{fig:setup}c). \newstuff{As shown below, in some cases, such a memory mechanism can be exploited to increase productivity. }
In the limit of $\lambda \to\infty$, equilibrium is reached within every inactive phase. In this {\newstuff {limit}}, internal memory of $\e_-$ is effectively erased during each inactive phase, and thus can no longer be exploited. \newstuff{Then, we say that} the strategy does not make use of any side information.

\subsubsection*{Productivity and information}
\label{subsec:experimental-measures}

\begin{figure*}[ht]
    \centering
    \includegraphics[width=2.05\mycolumnwidth]{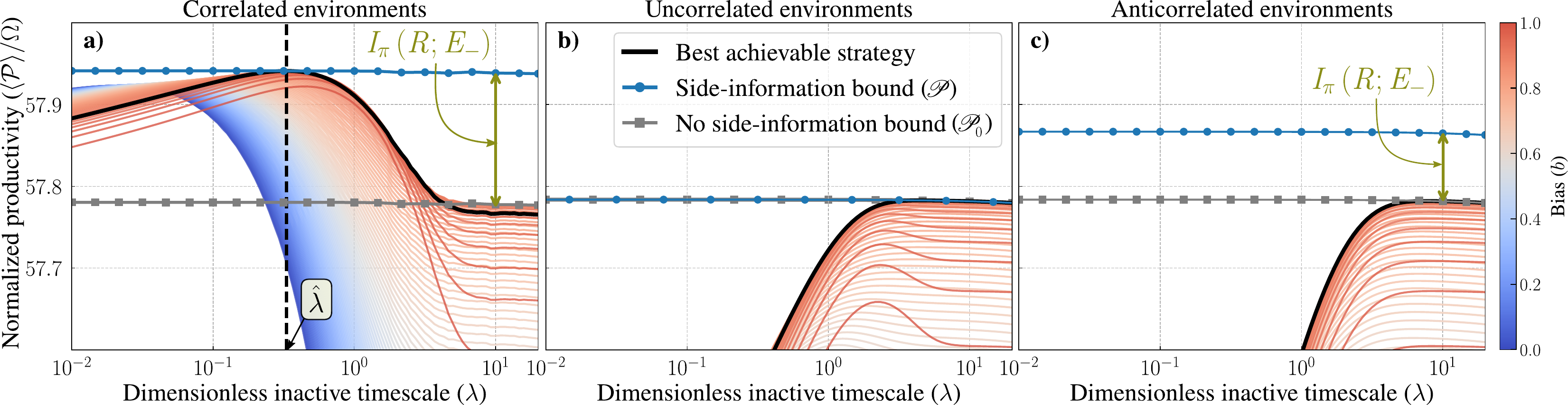}
    \caption{{\bf Information and productivity in \newstuff{simulated} photocatalytic replicator system}.
    We show normalized average productivity $\langle \pr \rangle/ \Omega$ as a function of  two control parameters: $\lambda$ (dimensionless inactive timescale from Eq.~\eqref{eq:dimensionless_timescale}; horizontal axis) and $\hp$ (\newstuff{bias in favor of replicator $\1_6$ due to exchange reactions}). \newnewstuff{Curves are colored according to the bias $\hp$, see colorbar on the right.} 
    Subplots {\bf (a)}, {\bf (b)} and {\bf (c)} correspond to temporally correlated ($c=0.175$), uncorrelated ($c=0$) and anticorrelated ($c=-0.1225$) environments, respectively. 
    \newnewstuff{The black curve in {\bf (a)} uses optimal bias $\hpopt=0.92$, obtained from Eq.~\eqref{eq:optimality0}, and in {\bf (b)} and {\bf (c)} uses the no-side-information bias $\hpopt_{\lambda \to\infty}=0.88$, obtained from Eq.~\eqref{eq:optimality0b}}. Blue lines:  productivity bounds $\prm$ from Eq.~\eqref{eq:second-main-result} with  side information about previous environment $E_-$. Gray lines:  productivity bound~$\prmO$ from Eq.~\eqref{eq:third-main-result} without side information.   
    For correlated environments~{\bf(a)}, the best achievable strategy has a finite \newnewstuff{optimal timescale $\lambdaopt$. We verify our analytical approximation of $\lambdaopt$ from Eq.~\eqref{eq:optimality0} (vertical dashed line) against the optimum found by numerics. We also verify Eq.~\eqref{eq:miexcess}:} the increase in the  productivity bound  is proportional to mutual information provided by side information.
    }
    \label{fig:prod-info}
\end{figure*}

We now calculate the average productivity for the photocatalytic replicator system. Recall from the last subsection that, due to incomplete relaxation during the inactive phase, the identity of the previous environment $\e_- \in \{\dim,\str\}$ can serve as side information for the current environment $\e\in\{\dim,\str\}$.  We then compute the average productivity as
\begin{align}
    \langle \pr\rangle= 
    \frac{\taua}{T}\sum_{\e,\e_-}  p_{\e,\e_-} \pr_{\e,\e_-}\,,
    \label{eq:avgenvPhotoCat}
\end{align}
where we used Eq.~\eqref{eq:avgenvScale} along with $\avgtau=\taua$ (from Eq.~\eqref{eq:avgtaudef}) and $\alpha=\taua/T$ (the fraction of time the reactor is open). The average steady-state productivity is given by
\begin{align}
    \langle \pr^*\rangle = \frac{\taua}{T} \big(p_{\scriptdim}\pr^*_{\scriptdim}+p_{\scriptstr}\pr^*_{\scriptstr}\big),\label{eq:ss-prod-photocat}
\end{align}
where $\pr^*_{\e}$ are obtained following the procedure discussed in the Methods. 
\newstuff{In Fig.~\ref{fig:setup}d, productivity $\pr(\tau)$ is obtained using Eq.~\eqref{eq:productivity-def} (black curve) and compared to Eq.~\eqref{eq:ss-prod-photocat} (dotted line).}

Following the expressions given in Eqs.~\eqref{eq:main-result},~\eqref{eq:pi-s}-\eqref{eq:omega-s} and~\eqref{eq:C}, for this setup we have:
\begin{align}
    \langle \pr \rangle &= \langle \pr^*\rangle-\gammaConst -
    \Omega  \, C_{\pi,q}(R\vert E_-)\\
\gammaConst&= \frac{\phi}{T}\sum_{\e,\e_-} \frac{p_{\e,\e_-}}{\eta_{r(\e)}} \ln \frac{X_{\e}^*}{X_{\inac}(0)}\,.\label{eq:exp_gamma_term}
\end{align}
Here,  \newstuff{$X_{\inac}(0)=X(0\vert y)$ indicates the total concentration at the end of the inactive phase, and it is well-approximated by $X_{\inac}(0)\approx \mu$ (see Methods).} 
We also have the decomposition of the cross-entropy,
\begin{align}
    C_{\pi,q}(R\vert E_-)&=H_{\pi}(R)-I_{\pi}(R;E_-)+ D(\pi_{R\vert E_-}\Vert q_{R\vert E_-}).
\end{align}
\newstuff{The mutual information term $I_{\pi}(R;E_-)$ vanishes when $c=0$ (uncorrelated environments) and it is strictly positive whenever $c \ne 0$ (correlated and anticorrelated environments).} 
Finally, following Eq.~\eqref{eq:pi-s}, \newstuff{and using} $\alpha/\avgtau=1/T$, we derive our optimal distribution $\pi$ as
\begin{align}
    &\pi_{6, \e_-} =  \frac{\phi}{\Omega T} \frac{p_{\scriptdim,\e_-}}{\eta_{6,\scriptdim}}, \quad 
     \pi_{3, \e_-} = \frac{\phi}{\Omega T} \frac{p_{\scriptstr, \e_-}}{\eta_{3,\scriptstr}}, \label{eq:pi_RY}
\end{align}
for $y=\e_-\in\{\smalldim,\smallstr\}$ and normalization constant
\begin{align}
\Omega = \frac{\phi}{T }\left(\frac{p_{\scriptdim}}{\eta_{6,\scriptdim}} + \frac{p_{\scriptstr}}{\eta_{3,\scriptstr}}\right). \label{eq:Omega_experiment}
\end{align}

\subsubsection*{Maximizing productivity}

As shown in Eq.~\eqref{eq:second-main-result}, productivity is maximized when the strategy $q_{R\vert E_-}$ matches the distribution $\pi_{R\vert E_-}$, at which point $\langle \pr\rangle=\prm$.  
However, in practice, one cannot always make $q_{R\vert E_-}$ equal to $\pi_{R\vert E_-}$ simply by varying the accessible control parameters chosen for this numerical experiment, namely $\{\hp,\lambda\}$.
Nonetheless, we can solve for the {\em best achievable strategy} given our set of controls. To do so, we explore how productivity varies with exchange bias $\hp$ and inactive dimensionless timescale $\lambda$.

There are two different ways of varying $\lambda$. 
For instance, one could keep \newstuff{$\kappa$} fixed and change $\taui$, the duration of the inactive phase. 
However, this affects the value of the cycle period $T$ and thus the average productivity in Eq.~\eqref{eq:avgenvPhotoCat}. 
In our example, we vary $\lambda$ by rescaling the overall \newstuff{exchange rate $\kappa$}, while keeping $\taui$ fixed. In practice, this could be accomplished by changing the temperature of the reactor, adding catalysts, etc.

In the Methods section, 
we derive the best achievable strategy by expressing $q_{R\vert E_-}$ as a function of $\{\hp,\lambda\}$, and then finding the values that minimize $C_{\pi,q}(R\vert Y)$. 
It turns out that the best achievable strategy depends on whether the environments are correlated, uncorrelated, or anticorrelated.

In Fig.~\ref{fig:prod-info}, we show numerical results for normalized time-averaged productivity $\langle \pr\rangle /\Omega$ (in dimensionless units) against the inactive timescale $\lambda$.  Colored curves correspond to different bias values $\hp$. To explore correlated, uncorrelated, and anticorrelated environments, we generate environments using a Markovian process with different transition probabilities between consecutive environments \newstuff{and fixed marginals, 
$p_{\scriptdim}=0.65=1-p_{\scriptstr}$}. 
Numerical values of productivity are obtained by averaging over many simulations.

For correlated environments~\newstuff{($c>0$, Fig.~\ref{fig:prod-info}a)}, the best achievable strategy has the following bias and inactive timescale: 
\newstuff{\begin{align}
\hpopt \approx  \frac{\pi_{6\vert \scriptstr}-\theta\left(1-\pi_{3\vert\scriptdim}\right)}{\pi_{6\vert \scriptstr}+\pi_{3\vert \scriptdim}-\theta},\;\lambdaopt\approx\ln\Big(\frac{1-\theta}{1-\pi_{6\vert \scriptstr}-\pi_{3\vert \scriptdim}}\Big),
\label{eq:optimality0}
\end{align}
where we introduced $\theta:={\phi}/{(\mu\eta_{3,\str})}$ for notational simplicity.} 
The derivation is provided in the Methods section. There, we also 
show that  $\pi_{6\vert \scriptstr}+\pi_{3\vert \scriptdim}<1$ when $c>0$, thus $\lambdaopt$ is well-defined. 
Under this strategy, 
productivity approaches the side information bound $\prm$ from Eq.~\eqref{eq:second-main-result}.

For uncorrelated~\newstuff{($c=0$, Fig.~\ref{fig:prod-info}b)} and anticorrelated~\newstuff{($c<0$, Fig.~\ref{fig:prod-info}c)} systems, 
the best inactive timescale diverges as $\lambdaopt\to \infty$. In this limit, the best bias is given by
\begin{align}
\hpopt_{\lambda\to\infty} 
=\pi_{6,\scriptdim}+\pi_{6,\scriptstr} = \pi_6\,.
\label{eq:optimality0b}
\end{align} \newstuff{We observe that memory decreases productivity in uncorrelated and anticorrelated environments, therefore  productivity is highest when the inactive timescale is long enough to erase all memory.  
This result appears counterintuitive at first, since in principle one can imagine a system that uses memory to increase productivity in anticorrelated environments. 
However, exploiting anticorrelations would generally  require the inactive phase to implement an `inverter'. This is not possible in this particular model, because the inactive phase dynamics are effectively one-dimensional, 
and such dynamics cannot implement an inverter (see Methods). 
Nonetheless, in more complex chemical networks, it is possible to overcome this limitation and  exploit anticorrelations to increase productivity.}

From the information-theoretical perspective, consider the bounds derived in Eq.~\eqref{eq:second-main-result}-\eqref{eq:third-main-result}. First, we note that productivity for all environments (Fig.~\ref{fig:prod-info}a,b and c) remains below the side-information bound $\prm$ (blue line). At large $\lambda$, all environments maximize productivity at $\prmO$. This is because, as discussed above, when $\lambda \to \infty$ all memory of the environmental fluctuations is erased from the initial concentrations, and the three scenarios become equivalent. On the other hand, at low values of $\lambda$, the system has little time to re-balance during inactivity, and there is not enough time to erase the memory of the previous environment, which hinders productivity in the cases of uncorrelated and anticorrelated environments. Although smaller, this effect is also present in correlated environments when $\lambda\ll \hat{\lambda}$, in which case the system does not erase enough memory.

\newstuff{However, in correlated environments, productivity exceeds the no-side-information bound $\prmO$ (gray line) at intermediate \newstuff{values of} $\lambda$. Here, the maximum productivity is non-monotonic and peaks at the predicted values of $\{\hpopt,\lambdaopt\}$ while reaching 
the side-information bound $\prm$. Hence, the gap exhibited between this peak and the large-$\lambda$ regime approximates the difference $\prm-\prmO = \Omega I_{\pi}(R;E_-)$. After normalization, this corresponds exactly to the mutual information between $R$ and $E_-$. 
Recall that, in our example, $R\equiv E$, thus $I_{\pi}(R;E_-)=I_{\pi}(E;E_-)\approx 0.13$ (nats) is the mutual information between consecutive environment states. This gap quantifies the amount of information about the environment that the system can potentially use to maximize productivity, and serves as an empirical signature of functional information in this chemical system.}  

In contrast, for uncorrelated environments, we observe that the two bounds coincide $\prm=\prmO$. This is because when $c=0$ there is no side information, $I_{\pi}(R;E_-)=0$. Finally, in anticorrelated environments, despite having $I_{\pi}(R;E_-)>0$, the system remains constrained by the no-side-information bound, \newstuff{which is a phenomenon particular to this case study, as discussed above.}

\begin{figure}[b]
    \centering
    \includegraphics[width=\mycolumnwidth]{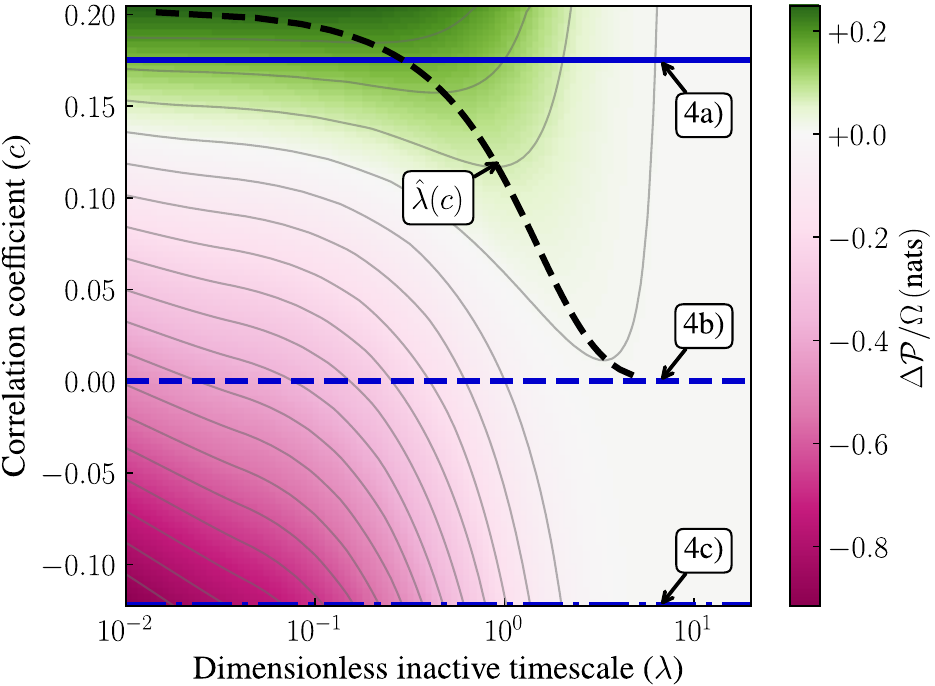}
\caption{{\bf Heat map of normalized productivity difference $\Delta \pr/\Omega $.}   The colorbar shows the values of $\Delta \pr/\Omega = \left(\max_\hp \langle \pr \rangle - \prmO\right)/\Omega$ for different values of $\lambda$ (dimensionless inactive timescale) and $c$ (correlation coefficient). Straight, dashed, and dot-dashed blue lines respectively correspond to subplots {\bf (a)}, {\bf (b)} and {\bf (c)} in Fig.~\ref{fig:prod-info}. Black dashed line indicates the values of the best achievable inactive timescale $\hat{\lambda}$ as a function of the coefficient $c$.}
    \label{fig:density-map}
\end{figure}

\newstuff{Finally,  Fig.~\ref{fig:density-map} is a heatmap that visualizes how productivity depends on the environmental correlations~$c$ and inactive timescale~$\lambda$. Here, we plot the (normalized) difference $\Delta \pr/\Omega$ between the maximum productivity $\max_\hp\langle \pr \rangle$ achievable at each value of $(\lambda,c)$, and the no-side-information bound $\prmO$. The bound $\prmO$ provides a useful baseline, since it can always be reached in the limit of long inactive time ($\lambda \to \infty$). Similarly to Fig.~\ref{fig:prod-info}, the values of $\Delta \pr/\Omega$ in Fig.~\ref{fig:density-map} quantify the informational contributions to productivity. They demonstrate an interesting interplay between environmental and internal memory. When there is internal memory (small $\lambda$, left side), $\Delta \pr/\Omega$  increases with environmental correlations (top), which can be exploited by the system, and decreases with environmental anticorrelations (bottom), for which memory only hurts productivity. Without internal memory (large $\lambda$, right side), the amount and structure of environmental correlations do not matter and $\Delta \pr/\Omega$ vanishes.   We also observe that the optimal timescale $\hat{\lambda}$ increases as $c \to 0$ from above, finally diverging to $\hat{\lambda} \to \infty$  for $c \le 0$.}

\section*{Discussion}
\label{sec:discussion}

In this paper, we established a connection between information and productivity in simple replicator systems exposed to fluctuating environments.
In particular, we showed that productivity has information-theoretic contributions arising from environmental uncertainty, side information, and the mismatch between the actual and optimal preparation strategies. 
We also derived an expression of the optimal strategy for maximizing productivity. 
Our approach extends existing ideas on informational limits on growth and selection to the realistic setting of simple chemical and biological replicators in flow reactors.


Our analysis contributes to the existing literature on information processing in biomolecular systems, including work on stochastic thermodynamics~\cite{bo2015thermodynamic,sartori2014thermodynamic,barato2014efficiency}, biochemical signaling~\cite{rhee2012application,govern2014energy}, regulation~\cite{tkavcik2011information,gregor2007probing}, and others. However, most previous research starts from the \emph{a priori} assumption that information-theoretic measures (Shannon entropy, channel capacity, etc.) are relevant for bounding functional performance (work extraction, signal transduction, etc.). In contrast, we start from a concrete setup (replicator dynamics in fluctuating environments) and then demonstrate that   
information-theoretic measures emerge as the relevant operational quantities.

Our analysis is closely related to the seminal work by Kelly on information and multiplicative growth~\cite{kelly1956new}. 
Originally operationalized in terms of gambling, Kelly's approach has since been used to study the relationship between information, fitness, and phenotypic variability in biology~\cite{haccou1995optimal,kussell2005phenotypic,donaldson2010fitness,kobayashi2015fluctuation,mayer2017transitions,marzenOptimizedBacteriaAre2018a,bernhardt2020life,rivoire2011value,moffett2022minimal,dinis2022pareto}. However, there are several important differences between previous Kelly-type analyses and our approach, making it more directly applicable to 
chemical and microbial replicators. 
First, we relate information to productivity in a finite flow reactor, not to unbounded exponential growth. 
Second, we demonstrate that both the `betting' and `gambling' stages of Kelly's setup can be implemented by a single continuous-time autonomous dynamical system (e.g., the dynamics of a population in a flow reactor). 
Finally, we demonstrate that a minimal replicator system can implement an internal memory and use it as a source of side information, without any explicit sensory mechanisms.



\newstuff{There is also an interesting connection to the concept of `substitutional load' in evolutionary biology~\cite{haldaneCostNaturalSelection1957,kimura1961natural,ewensRemarksSubstitutionalLoad1970, mcgee2022cost}. Consider a biological population containing two alleles, one of which is initially rare but has higher fitness. 
Substitutional load refers to the cost of replacing the less fit allele with the fitter one by the process of natural selection, where cost is measured as decreased population fitness (proportional to the additional deaths needed to remove the less fit organisms). Kimura showed that substitutional load can be equivalently expressed as the negative logarithm of the initial proportion of the fittest allele~\cite{kimura1961natural}. Similarly, our setup considers the dynamics of a fitter replicator as it heads towards fixation in a favorable environment. 
In analogy with Kimura's substitutional load, the second (logarithmic) term in Eq.~\eqref{eq:prod-result-2} can be understood as the productivity cost  incurred in order to 
increase the concentration of the winning replicator. Note that Kimura assumed a fixed population size, whereas our analysis allows the total replicator concentration to vary over time.} \newnewstuff{(See also Ref.~\cite[pp.~160-164]{crow1970genetic} for a general discussion of substitutional load and other costs in systems with changing population sizes.)}

To contextualize our work, it is useful to distinguish three different types of selection or adaptation. 
The first, and most familiar, type is natural selection in a fixed environment. This occurs during our active phases, where the winning replicator is selected via competitive exclusion.
\newstuff{The second type is competition and selection among  {replicator networks}, i.e., networks of replicating species linked by exchange reactions.  
\newnewstuff{This may be seen as a form of natural selection,  where the units of selection are entire replicator networks rather than replicator species.}  
A possible prebiotic realization may be imagined in porous systems like hydrothermal vents~\cite{baaske2007extreme}, where each pore acts as a microscopic flow reactor containing different competing replicator networks. In such cases, productivity (outflow of replicators) serves as a measure of fitness, since it governs the ability of a replicator network to colonize new pores. This allows us to imagine selection for  networks whose exchange reactions implement better strategies, thus leading to higher productivity (all else being equal) over many cycles of inactive and active phases. 
We note that, although the productivity contribution from information-theoretic terms may not be very large (only a few percent in Fig.~\ref{fig:prod-info}a), the effect on selection between networks can be very significant.} 
The third and final type of adaptation is `ontogenetic' 
improvement of a single replicator type or replicator network via autonomous (i.e., internal) mechanisms. For example, one may imagine a chemical system that performs associative learning, as illustrated in Ref.~\cite{bartlett2022provenance} using a Gray-Scott reaction-diffusion model. In the setting of replicators in fluctuating environments, one might imagine a mechanism that allows for  slow modifications of internal variables that affect exchange kinetics, in this way leading to adaptation of the strategy  to environmental statistics. In this work, we focus on productivity and information in fixed replicator systems with fixed strategies (e.g., fixed exchange reactions), so we do not consider any kind of ontogenetic adaptation.  Studying this type of adaptation in replicator networks is an interesting direction for future research.

In addition to our theoretical analysis, we illustrated our results on a  
model of real-world photocatalytic replicators~\cite{liu2024light}. We demonstrated that this chemical system, when exposed to a fluctuating environment, can implement a strategy and maintain an internal memory of previous environments as side information, without requiring dedicated  sensory mechanisms.  Finally, we showed that productivity provides a signature of information flow in this  plausible experimental setup. 
In our analysis of photocatalytic replicators,  productivity depends both on the growth rates of the replicators as well as the (slower) exchange reactions that lead to re-balancing of replicator concentrations.  We argue that these exchange reactions may be interpreted as performing `information processing', in the sense that they map input states (concentrations at the end of the previous active phase) to  output states (concentrations at the beginning of the next active phase) in a way that affects  functional consequences (productivity).  From this perspective, our information-theoretic decomposition of productivity quantifies the efficacy of the network's information processing, based on the alignment between actual environmental statistics and those implicitly encoded in the strategy. 
In this way, our analysis provides a concrete illustration of how 
simple chemical systems can store and process information in fluctuating conditions. 
We note that we focus only on first-order internal memory, where only the previous environment is tracked. Future work may consider chemical systems that maintain higher-order memories, allowing for the tracking and processing of longer environmental histories.




We finish by mentioning two other directions for future research. 

\newstuff{First, our theoretical analysis was based on a particular model of replicator dynamics: during active phases, the replicators grow with first-order kinetics from a single substrate reactant, until only a single fittest type remains at nonnegligible concentrations.  
Such dynamics describe many autocatalytic reaction schemes, and even some autocatalytic networks given an appropriate separation of timescales~\cite{kolchinsky2024thermodynamics}, but other dynamical regimes are also possible. 
Future work may extend our analysis to more complex replicator networks, including those with multiple reactants, non-first-order growth, non-negligible exchange reactions during active phases, incomplete relaxation to steady state, and/or coexistence of multiple winner replicators. 
To make progress,  it may be possible to adapt techniques from Ref.~\cite{donaldson2010fitness}, which analyzed the ``fitness value of information'' in the presence of co-existing winning types.}

Second, here we focused entirely on deterministic chemical systems, which is justified when concentrations are large enough so that thermal fluctuations can be ignored. Extending our formalism to stochastic chemical reactions may shed light on how thermal noise influences the relationship between information and productivity. It may also suggest interesting connections between our approach and  recent results from nonequilibrium and stochastic thermodynamics. 

\section*{Methods}

\subsection*{Derivation of main results}

\subsubsection*{{Average productivity as a time average, Eq.~\eqref{eq:avgenvScale}}}
\label{app:average-productivity}

\newstuff{Here, we show that $\langle \pr\rangle$, as written in Eq.~\eqref{eq:avgenvScale}, can be regarded as the system's productivity averaged over time. 
First, we consider a random variable with outcomes $\{(\e,y)\}_{\e,y} \cup \{\inac\}$, which refer to different types of time periods. Outcome $(\e,y)$ represents an active time period under environment $\e$ and preparation $y$. The outcome $\inac$ represents an inactive time period. The probability distribution over outcomes, written  $v$, is defined in terms of the fraction of time that the system spends in each type of time period:
\begin{align}
    v_{\e,y}&= \alpha p_{\e,y}\tau_{\e}/\avgtau \\
    v_{\inac} &= 1-\alpha
\end{align}
where $p_{\e,y}$ is the frequency of environment $\e$ and preparation $y$, $\tau_\e$ is length of time of environment $\e$, and $\avgtau   := \sum_{\e} p_{\e}\tau_{\e}$ is the average time length of an environment,  Eq.~\eqref{eq:avgtaudef}. 
Recall also that $\alpha$ denotes the fraction of time spent in active phases, so $1-\alpha$ corresponds to the fraction within inactive phases. It is easy to verify that $v$ is nonnegative and sums to 1, hence it is a valid probability distribution.}

\newstuff{Using this definition, $v$ can be used to calculate time averages over observables. In particular, consider the productivity observable $\pr$, where $\pr_{\e,y}$ is the productivity under  active environments $\e$ and preparation $y$ and $\pr_{\inac}=0$ during inactive phases (when the system is closed). Then, it is easy to show that $\langle \pr \rangle$ from Eq.~\eqref{eq:avgenvScale} is the expectation of $\pr$ under $v$.
}

\subsubsection*{Derivation of average productivity using cross-entropy, Eq.~\eqref{eq:main-result}}
\label{app:main-result-derivation}
Using Eqs.~\eqref{eq:prod-e-y}-\eqref{eq:def-ss-av-prod}, we write the average productivity as
\begin{align}
    \langle \pr \rangle &= \langle \pr ^*\rangle + \frac{\alpha}{\avgtau}\sum_{\e,y}p_{\e, y} \multATNTe\left[\ln q_{r(\e)\vert y} + \ln \frac{X(0\vert y)}{X_{\e}^*}\right]\nonumber\\
    &= \langle \pr ^*\rangle -\gammaConst + \frac{\alpha}{\avgtau}\sum_{\e,y}p_{\e, y} \multATNTe \ln q_{r(\e)\vert y} \,,
    \label{eq:zz99dd}
\end{align}
where in the second line we used the definition of $\gammaConst$ from Eq.~\eqref{eq:def-gamma}. 
Next, we recall 
the information-theoretic expression of cross-entropy from Eq.~\eqref{eq:C-0},
\begin{align}
    C_{\pi,q}(R\vert Y):=-\sum_{r,y}\pi_{r,y}\ln q_{r\vert y}\,.\label{eq:app-Cost}
\end{align}
Plugging in the  definition of $\pi_{r,y}$ from Eq.~\eqref{eq:pi-s} gives 
\begin{align*}
    C_{\pi,q}(R\vert Y)&= -\frac{\alpha}{\Omega\avgtau}\sum_{r,y}\sum_{\e: r(\e)=r}p_{\e,y}\multATNTe\ln q_{r\vert y} \\
    &=-\frac{\alpha}{\Omega\avgtau}\sum_{\e,y}p_{\e,y}\multATNTe\ln q_{r(\e)\vert y}\,.
\end{align*}
Combining with Eq.~\eqref{eq:zz99dd} gives $\langle \pr \rangle =  \langle \pr^*\rangle -\gammaConst- \Omega \, C_{\pi,q}(R\vert Y)$, which is Eq.~\eqref{eq:main-result}.

\subsection*{Photocatalytic replicator model}

\subsubsection*{Active phase}
\label{app:active}

During active phases, the system evolves according to:
\begin{align}
    \frac{da_{\e}}{dt}=& \ (\mu-a_{\e})\phi-\left(\eta_{6,\e}x_{6,\e}+\eta_{3,\e}x_{3,\e}\right)a_{\e} ,\label{eq:Act-dyn-a}\\
    \frac{dx_{6,{\e}}}{dt}=&\left(\eta_{6,\e}a_{\e}-\phi\right)x_{6,{\e}}+\newstuff{\kappa \hp}x_{3,\e}-\newstuff{\kappa(1-\hp)}x_{6,\e} ,\label{eq:Act-dyn-u}\\
    \frac{dx_{3,{\e}}}{dt}=&\left(\eta_{3,\e}a_{\e}-\phi\right)x_{3,\e}+\newstuff{\kappa(1-\hp)}x_{6,\e}-\newstuff{\kappa \hp}x_{3,\e},\label{eq:Act-dyn-v}
\end{align}
\newstuff{where we recall $a_{\e}$, $x_{6,\e}$ and $x_{3,\e}$ denote the concentrations of the reactant ($\1$) and the hexameric ($\1_6$) and trimeric ($\1_3$) replicators. We also} recall that we prepare the system such that $S(0)=S^*=\mu$ (for example, by letting the system flow at $\phi$ before starting the experiment). Hence, at all times we have that
\begin{align}
    S=a_{\e}+x_{6,\e}+x_{3,\e}=a_{\e}+X_{\e}=\mu.\label{eq:constant-solute-condition}
\end{align}
In our setup, initial conditions for an active phase are given by the final concentration values from the previous inactive state, which we discuss next. 
We solve equations~\eqref{eq:Act-dyn-a},~\eqref{eq:Act-dyn-u} and~\eqref{eq:Act-dyn-v} numerically using the Runge-Kutta method. 
As an example, Fig.~\ref{fig:active-phase-trajs}a shows the computed trajectories for $x_{6,\scriptstr}(t)$ and $x_{3,\scriptstr}(t)$ under strong light, $\e=\smallstr$.

When needed, steady-state values for $x^*_{6,\scriptdim}$ and $x^*_{3,\scriptstr}$ can be estimated using Eq.~\eqref{eq:xr-solution}.

\begin{figure}
    \centering
    \includegraphics[width=\mycolumnwidth]{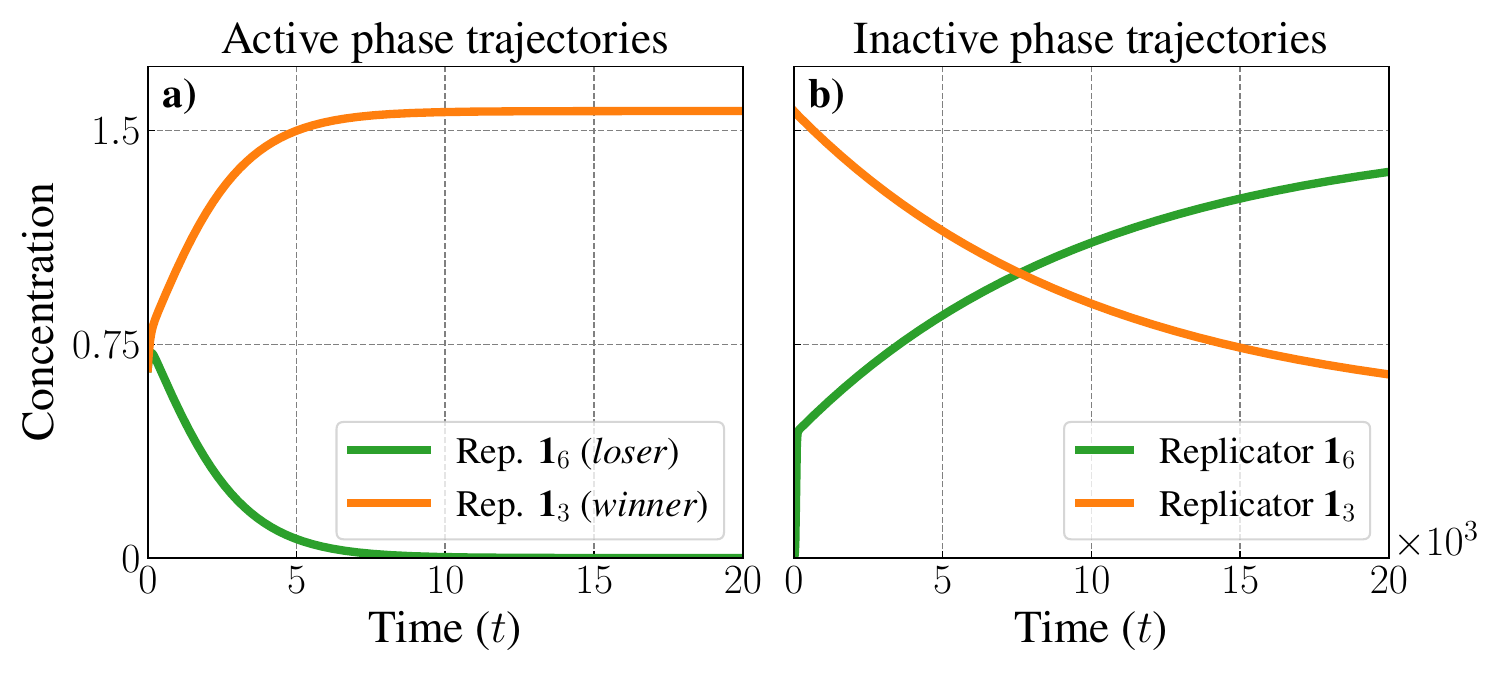}
    \caption{{\bf Concentration trajectories}. For the set of parameters introduced in  our Application to a real-world system: {\bf (a)} Concentration trajectories during an active phase with $\e=\smallstr$, with $a(0)=x_6(0)=x_3(0)=\mu/3$. {\bf (b)} Concentration trajectories during an inactive phase proceeding from the endpoints of the trajectories in the left panel. The timescale difference between the subplots reflects the dominant rates in each phase.}\label{fig:active-phase-trajs}
\end{figure}

\subsubsection*{Inactive phase}
\label{app:inactive}

During inactive phases, the system evolves according to:
\begin{align}
    \frac{da_{\inac}}{dt}&=-\eta_{6,\inac}a_{\inac}x_{6,\inac},\label{eq:A-dyn-a}\\ 
    \frac{dx_{6,\inac}}{dt}&=\eta_{6,\inac}a_{\inac}x_{6,\inac}+\kappa \hp x_{3,\inac}-\kappa(1-\hp) x_{6,\inac} ,\label{eq:A-dyn-u}\\
    \frac{dx_{3,\inac}}{dt}&=\kappa(1-\hp) x_{6,\inac}-\kappa \hp x_{3,\inac} \,.\label{eq:A-dyn-v}
\end{align}

Solving the above ODEs in closed form is challenging. However, for our parameter values, we may approximate them with a simpler and solvable system. 
During the inactive phase, exchange reactions are much slower than autocatalytic replication of $\1_6$, because $\kappa \ll \eta_{6,\inac}\mu$. Hence, any reactant is quickly converted into $\1_6$. This effect can be seen in the inactive phases shown in Fig.~\ref{fig:setup}c or in the initial increase in $\1_6$ concentration shown in Fig.~\ref{fig:active-phase-trajs}b.  Therefore, we may assume that $a_\inac\approx0$ throughout. Moreover, since  $a_{\inac}+X_{\inac}=S^*=\mu$ holds at all times, we have that 
\begin{align}
x_{6,\inac}+x_{3,\inac}\approx \mu\,,\label{eq:x6_x3_mu_condition}
\end{align}
where we used $X_{\inac}=x_{6,\inac}+x_{3,\inac}$. We now consider the system of equations
\begin{align}
    \frac{dx_{6,\inac}}{dt}&\approx\kappa \hp x_{3,\inac}-\kappa(1-\hp) x_{6,\inac} ,\label{eq:A-dyn-u-approx}\\
    \frac{dx_{3,\inac}}{dt}&\approx \kappa(1-\hp) x_{6,\inac}-\kappa \hp x_{3,\inac} \,.\label{eq:A-dyn-v-approx}
\end{align}
The dynamical system above is solved by
\begin{align}
    x_{6,\inac}(t|\e_-)&\approx \hp \mu + \left[x_{6,\inac}(0|\e_-)-\hp \mu\right]e^{-\kappa t}, \label{eq:A-u-inact-evol} \\
    x_{3,\inac}(t|\e_-)&\approx (1-\hp) \mu + \left[x_{3,\inac}(0|\e_-)-(1-\hp) \mu\right]e^{-\kappa t}, \label{eq:A-v-inact-evol}
\end{align}
where we plugged in Eq.~\eqref{eq:x6_x3_mu_condition}.
For the initial conditions, we assume that all the reactant is quickly turned into $\1_6$, thus
\begin{align}
    x_{3,\inac}(0|\e_-)&\approx x_{3,\scriptstr}^*\delta_{\e_-\scriptstr}\,,\\
    x_{6,\inac}(0|\e_-)&\approx \mu - x_{3,\inac}(0|\e_-)=\mu-x_{3,\scriptstr}^*\delta_{\e_-\scriptstr}\,.
\end{align}
We later use Eq.~\eqref{eq:xr-solution} to approximate $x_{3,\scriptstr}^*\approx \mu -\phi/\eta_{3,\scriptstr}$.

The term $e^{-\kappa t}$ that appears in Eq.~\eqref{eq:A-u-inact-evol}-\eqref{eq:A-v-inact-evol} controls the tradeoff between the initial (square brackets) and inactive steady-state concentrations. At the end of the inactive phase $t=\taui$,  $e^{-\kappa t}\big|_{t=\taui}=e^{-\lambda}$, hence the dimensionless inactive timescale $\lambda=\kappa \taui$ controls this tradeoff. Fig.~\ref{fig:active-phase-trajs}b shows typical trajectories for $\e_-=\smalldim$.

We now note that, due to the condition~\eqref{eq:x6_x3_mu_condition},  the dynamics over $x_{6,\inac},x_{3,\inac}$ is effectively one-dimensional. Such dynamics cannot implement an `inverter' that maps higher initial $x_6$ to higher final $x_3$ and vice versa. This is because an inverter would require the difference $x_6-x_3$ to evolve in different directions depending on initial conditions, which is impossible in a one-dimensional system (see also Ref.~\cite{owen2019number}). For this reason, the system cannot exploit environmental anticorrelations to increase productivity.

Finally, we define the concentrations at the end of an inactive phase as functions of $\{\hp,\lambda\}$. We condition on the state of the previous active phase by substituting for $t=\taui$ in Eqs.~\eqref{eq:A-u-inact-evol} and~\eqref{eq:A-v-inact-evol}. 
In particular, it will be useful to consider the endpoints of the inactive phase trajectories of the concentrations of $\1_6$ and $\1_3$ when preceded by light intensity $\e_-=\str$ and $\e_-=\dim$, respectively. For convenience, we introduce the notation $x^{\e_-}_{i,\inac}(\hp,\lambda):=x_{i,\inac}(\taui \vert \e_-)$ for $i\in\{6,3\}$.  
We also introduce the dimensionless quantity \begin{align}
    \theta &:= \frac{\phi}{\mu\eta_{3,\scriptstr}}\label{eq:smthetadef}
\end{align}
to quantify the ratio of dilution to maximal growth of replicator $i=3$ (washout occurs when $\theta > 1$).
Then, 
\begin{align}
    x^{\scriptstr}_{6,\inac}(\hp,\lambda)&\approx \left[\hp  -\left(\hp - \theta\right) e^{-\lambda} \right]\mu\,,
    \label{eq:x1_inac_dim}\\
    x^{\scriptdim}_{3,\inac}(\hp,\lambda)&\approx (1-\hp)\left(1-e^{-\lambda}\right)\mu\,.
    \label{eq:x2_inac_dim}      
\end{align}

\subsubsection*{Best achievable strategy}
\label{app:opt}

In order to study the best achievable strategy, we recall from our main result, Eq.~\eqref{eq:main-result}, that all the dependence on the strategy $q$ is encoded in our information-theoretic cost $C_{\pi,q}(R\vert Y)$, given in Eq.~\eqref{eq:C}. 

In our real-world example introduced above, 
the parameters that control $q$ are $\{\hp,\lambda\}$, i.e., $q=q(\hp,\lambda)$. In general, there may not be values of $\{\hp,\lambda\}$ such that $q$ equals $\pi$ and thus achieves maximum productivity. However, we can still optimize the contribution in Eq.~\eqref{eq:C} in each case. 

Let us write the strategy, conditional on the previous environment states, as fractions of respective concentrations evaluated at the end of the inactive phase:
\begin{align}
    q_{6\vert \scriptstr}(\hp,\lambda) &= \frac{x^{\scriptstr}_{6,\inac}(\hp,\lambda)}{\Xstrin(\hp,\lambda)}\newstuff{\approx \hp -\left(\hp-\theta\right)e^{-\lambda}}, \label{eq:q1-str-approx}\\
    q_{3\vert \scriptdim} (\hp,\lambda) &= \frac{x^{\scriptdim}_{3,\inac}(\hp,\lambda)}{\Xdimin(\hp,\lambda)}\newstuff{\approx(1-\hp) \left(1-e^{-\lambda}\right)}.\label{eq:q2-dim-approx}
\end{align}
with $q_{3\vert \scriptstr}=1-q_{6\vert \scriptstr}$ and $q_{6\vert \scriptdim}=1-q_{3\vert \scriptdim}$,  where we used Eqs.~\eqref{eq:x1_inac_dim}-\eqref{eq:x2_inac_dim} \newstuff{and $\theta$ from Eq.~\eqref{eq:smthetadef}}.

We now use the expressions above to solve for the best achievable strategy. As described above, the best strategy is obtained by minimizing $C_{\pi,q}(R|Y)$ 
with respect to the bias $\hp$ and the dimensionless inactive timescale $\lambda$. First, we write the cross-entropy term as a function of $\{\hp,\lambda\}$ by plugging in the approximations~\eqref{eq:q1-str-approx}-\eqref{eq:q2-dim-approx} into Eq.~\eqref{eq:C-0} and then use the conditional distribution $\pi_{R\vert Y}$, which yields
\begin{align}
    C_{\pi,q}(R\vert Y) \approx &- \pi_{6\vert \scriptdim}\pi_{\scriptdim}\ln \left[1-(1-\hp)\left(1-e^{-\lambda}\right)\right]\nonumber\\
    &-\pi_{6\vert \scriptstr}\pi_{\scriptstr}\ln\left[\hp-(\hp-\theta)e^{-\lambda}\right]\nonumber\\
    &-\pi_{3\vert \scriptdim}\pi_{\scriptdim}\ln\left[(1-\hp)\left(1-e^{-\lambda}\right)\right]\nonumber\\
    &-\pi_{3\vert \scriptstr}\pi_{\scriptstr}\ln\left[1-\hp+(\hp-\theta)e^{-\lambda}\right].\label{eq:approxC}
\end{align}
Next, we use $\pi_{6\vert \scriptdim}=1-\pi_{3\vert\scriptdim}$, $\pi_{3\vert \scriptstr}=1-\pi_{6\vert\scriptstr}$, and $\pi_{\scriptstr}+\pi_{\scriptdim}=1$. We find 
the optimal bias and dimensionless timescale by taking derivatives and setting them to zero: 
\begin{align*}{\partial_\hp}C_{\pi,q}(R\vert Y)\vert_{\hp=\hpopt} = 0\,,\quad{\partial_\lambda}C_{\pi,q}(R\vert Y)\vert_{\lambda=\lambdaopt} = 0\,.
\end{align*}
With a bit of algebra (or with software like Mathematica), this system of equations can be solved to give 
\newstuff{
\begin{align}
    \begin{aligned}
\hpopt &\approx  \frac{\pi_{6\vert \scriptstr}-\theta\left(1-\pi_{3\vert\scriptdim}\right)}{\pi_{6\vert \scriptstr}+\pi_{3\vert \scriptdim}-\theta},\\
\lambdaopt&\approx-\ln\left(\frac{1-\pi_{6\vert \scriptstr}-\pi_{3\vert \scriptdim}}{1-\theta}\right).
\end{aligned}
\end{align}
}
\newstuff{When there is no washout ($\theta<1$),} this solution is not valid for uncorrelated and anticorrelated systems, for which $\pi_{6\vert \scriptstr}+\pi_{3\vert \scriptdim}\ge 1$,
because the critical point is outside of the valid parameter region $(\hp,\lambda)\in[0,1]\times \mathbb{R}^+$. Therefore, for uncorrelated and anticorrelated systems, the minimum of $C_{\pi,q}(R\vert Y)$ must either be achieved at the boundaries ($\hp=0$ or $\hp=1$ and $\lambda=0$), or not achieved so that $C_{\pi,q}(R\vert Y)$ continually decreases as $\lambda\to\infty$. However, given Eq.~\eqref{eq:approxC}, we note that for $\hp\to 0$, $\hp\to 1$, and $\lambda\to0$, $C_{\pi,q}(R\vert Y)\to +\infty$ due to the $\ln(0)$ terms. Hence, the minimum cannot be achieved at the boundaries, which means that the best timescale for uncorrelated and anticorrelated environments diverges,
\begin{align}
    \lambda \to \infty \,.
\end{align}
Moreover, by studying $\lim_{\lambda\to\infty}C_{\pi,q}(R\vert Y)$ as a function of $\hp$ and minimizing, we find:
\begin{align}
    \hpopt_{\lambda\to \infty}=\pi_{6,\scriptstr}+\pi_{6,\scriptdim}=\pi_6.
\end{align}

\subsubsection*{Conditions on $p$ and $\pi$ in correlated vs. anticorrelated environments}
\label{app:correlated}

\newcommand\pI{(i)}
\newcommand\pII{(ii)}
\newcommand\pIII{(iii)}
\newcommand\pIV{(iv)}

\newcommand\pIs{(i)}
\newcommand\pIIs{(ii)}
\newcommand\pIIIs{(iii)}
\newcommand\pIVs{(iv)}

Here we derive conditions on $p$ and $\pi$ in correlated vs. anticorrelated environments. We will make repeated use of the following general result. 
\begin{prop}
\label{prop:1}
Let $\omega_{AB}$ be a joint probability distribution over two binary random variables: $A$ with outcomes $\{1,2\}$ and $B$ with outcomes $\{\downarrow,\uparrow\}$. If the marginals $\omega_A$ and $\omega_B$ have full support, the following four statements are equivalent:
\begin{align*}
\text{\pIs\;\;} \omega_{1,\downarrow} &> \omega_{1}\omega_{\downarrow}\quad & & \text{\pIIIs\;\;} \omega_{1\vert \downarrow}+\omega_{2\vert \uparrow} > 1\\
\text{\pIIs\;\;} \omega_{2,\uparrow} &> \omega_{2}\omega_{\uparrow} \quad& &
\text{\pIVs\;\;} \omega_{\downarrow\vert 1}+\omega_{\uparrow\vert 2} > 1
\end{align*}
\end{prop}
\begin{proof}
    First, we rewrite both sides of \pI{} as
\begin{align}
    (1-\omega_{2}-\omega_{\uparrow}+\omega_{2,\uparrow})>(1-\omega_{2})(1-\omega_{\uparrow}).
\end{align}
Expanding and canceling terms shows equivalence with \pII{}. 

To show equivalence with \pIII{}, we divide both sides of \pI{} by $\omega_\downarrow >0$ to 
give $\omega_{1\vert \downarrow} > \omega_{1}$ and both sides 
of \pII{} by $\omega_\uparrow >0$ to give $\omega_{2\vert \uparrow} > \omega_2$. Adding both inequalities and using $\omega_1+\omega_2=1$ implies \pIII{}. To show the reverse implication, observe that \pIII{} can only be true if at least one of {\pI{},\pII{}} is true. However, since \pI{} and \pII{}  are equivalent, they must both be true when \pIII{} holds.

Equivalence of \pIV{} is derived in a similar way, by dividing \pI{} by $\omega_1 >0$ and \pII{} by $\omega_2 > 0$, then adding the inequalities.
\end{proof}

Eq.~\eqref{eq:twostatecov} follows from Prop.~\ref{prop:1} by taking $A=E$ ($1=\dim,2=\str$) and $B=E_-$ ($\downarrow=\dim,\uparrow=\str$), then using the equivalence of \pI{} and \pII{}. 

Next, we use Prop.~\ref{prop:1} to prove the equivalence between the statement 
\begin{align}
p_{\scriptdim,\scriptdim} > p_{\scriptdim}p_{\scriptdim},\label{eq:p_corr_ineq}
\end{align}
\newstuff{which corresponds to $c>0$, given its definition in Eq.~\eqref{eq:twostatecov}}, and the statement
\begin{align}
    \pi_{6\vert \scriptstr}+\pi_{3\vert \scriptdim}<1\,,\label{eq:biasvalid}
\end{align}
which is used below to derive the best achievable strategy. 
Eq.~\eqref{eq:p_corr_ineq} can be put in the form of Prop.~\ref{prop:1}\pI{}, which is equivalent to Prop.~\ref{prop:1}\pIV{},
\begin{align}
p_{E_-=\scriptdim\vert E=\scriptdim} + p_{E_-=\scriptstr\vert E=\scriptstr} > 1\,.
\label{eq:appn11}
\end{align}
Next, we rearrange Eq.~\eqref{eq:pi_RY} to show
\begin{align*}
    \pi_{6}=\frac{\phi}{\Omega T}\frac{p_\scriptdim}{\eta_{6,\scriptdim}} &\implies \pi_{\e_-\vert 6}=p_{E_-=\e_-\vert E=\scriptdim}\\
        \pi_{3}=\frac{\phi}{\Omega T}\frac{p_\scriptstr}{\eta_{3,\scriptstr}} &\implies \pi_{\e_-\vert 3}=p_{E_-=\e_-\vert E=\scriptstr}\,,
\end{align*}
therefore Eq.~\eqref{eq:appn11} is equivalent to
\begin{align}
 \pi_{\scriptdim \vert 6} +  \pi_{\scriptstr\vert 3} > 1\,.
\label{eq:appn12}
\end{align}
Third, we apply Prop.~\ref{prop:1} to the joint distribution $\pi_{RY}$, taking $A=R$  and $B=E_-$ ($\downarrow=\dim,\uparrow=\str$). We then have the equivalence of Prop.~\ref{prop:1}\pIV{} in Eq.~\eqref{eq:appn12} and Prop.~\ref{prop:1}\pIII{}:
\begin{align}
\pi_{6\vert \scriptdim}+\pi_{3\vert \scriptstr}>1\,.
\label{eq:appn13}
\end{align}
Eq.~\eqref{eq:biasvalid} follows from Eq.~\eqref{eq:appn13} and $\pi_{6\vert \scriptdim}+\pi_{3\vert \scriptdim}=\pi_{6\vert \scriptstr}+\pi_{3\vert \scriptstr}=1$.

\subsubsection*{\newstuff{Environment transition probabilities and bounds on the correlation $c$}}
\label{app:map-to-c-x}

\newstuff{Here we compute the transition probabilities between environments $p_{\e,\e_-}$ as a function of the marginal probability $p_{\smalldim}$ and correlation $c$. 
First, we use $p_{\e,\e_-} = p_{\e|\e_-}p_{\e_-}$ to write:
\begin{align*}
 p_{\scriptdim,\scriptdim}  = \left(1-p_{\scriptstr|\scriptdim}\right)p_{\scriptdim} \;\; \textrm{and}\;\; 
p_{\scriptstr,\scriptstr}  = \left(1-p_{\scriptdim|\scriptstr}\right)(1-p_{\scriptdim})
\end{align*}
Using Eq.~\eqref{eq:twostatecov}, we can write $c$ in two ways:
\begin{align}
    c&=p_{\scriptdim,\scriptdim}-p_{\scriptdim}^2 = p_{\scriptdim}(1-p_{\scriptdim}-p_{\scriptstr|\scriptdim})\\
c &=p_{\scriptstr,\scriptstr}-p_{\scriptstr}^2= (1-p_{\scriptdim})(p_{\scriptdim}-p_{\scriptdim|\scriptstr}).
\end{align}
From the above, it is possible to obtain
\begin{align}
    p_{\scriptstr|\scriptdim} &= 1-p_{\scriptdim} -\frac{c}{p_{\scriptdim}},\label{eq:ssmap0} \\
    p_{\scriptdim|\scriptstr} &= p_{\scriptdim}-\frac{c}{1-p_{\scriptdim}}.\label{eq:ssmap1}
\end{align}
which completes our map.
}

\newstuff{Finally, we note that $c$ is bounded as $c_{\min}\leq c \leq c_{\max}$. These bounds can be derived from constraints on the conditional probability, $0\leq p_{\scriptstr|\scriptdim}\leq 1$ and $0\leq p_{\scriptdim|\scriptstr}\leq 1$. Using Eq.~\eqref{eq:ssmap0}, we obtain
\begin{align}
    -p_{\scriptdim}^2&\leq c \leq p_{\scriptdim}(1-p_{\scriptdim})\,.
\end{align}
Using Eq.~\eqref{eq:ssmap1}, we obtain  
\begin{align}
    -(1-p_{\scriptdim})^2&\leq c \leq p_{\scriptdim}(1-p_{\scriptdim})
\end{align}
Choosing the most restrictive lower bound yields
\begin{align}
    c_{\min} = \max \{-p_{\scriptdim}^2,-(1-p_{\scriptdim})^2\}, 
\end{align}
and the upper bound is determined by $c_{\max}=p_{\scriptdim}(1-p_{\scriptdim})$.}

\subsubsection*{Data availability}

Data sharing is not applicable to this article as no datasets were generated in the study. All numerical analyses use standard implementations of ODEs.

\subsubsection*{Contributions}
J.P. and A.K. conceptualized the project, performed the theoretical development, simulations, data analysis, and wrote the first draft of the manuscript. 
J.P., D.S., G.G., A.F., and A.K. contributed to discussion of results, editing, and revisions. A.K. supervised the project. 

\subsubsection*{Competing interests}
The authors declare no competing interests.

\vspace{1pt}

\begin{acknowledgments}
This project was supported by Grant No.~62417 from the John Templeton Foundation. The opinions expressed in this publication are those of the authors and do not necessarily reflect the views of the John Templeton Foundation. AK was partly supported by the European Union's Horizon 2020 research and innovation programme under the Marie Sk{\l}odowska-Curie Grant Agreement No.~101068029.
\end{acknowledgments}

\bibliographystyle{ieeetr}
\bibliography{bibliography}

\end{document}